\documentclass[a4paper,UKenglish,cleveref, autoref, thm-restate]{lipics-v2019}

\usepackage[]{algorithm2e}
\newcommand{\enquote}[1]{``#1''}
\newcommand{\norm}[1]{\left\lVert #1 \right\rVert}
\newcommand{\reals}{\mathbb{R}}
\newcommand{\eps}{\varepsilon}
\newcommand{\naturals}{\mathbb{N}}
\newcommand{\rationals}{\mathbb{Q}}
\newcommand{\ipsim}{ip{\hbox{-}}sim}

\newcommand{\CS}{\mbox{Centroid-SLSH}\xspace}
\newcommand{\RS}{\mbox{Repeat-SLSH}\xspace}
\newcommand{\ES}{\mbox{Exhaustive-SLSH}\xspace}
\newcommand{\WES}{\mbox{Weighted exhaustive-SLSH}\xspace}
\newcommand{\SL}{\mbox{Shrink-lift-SLSH}\xspace}
\newcommand{\sa}{\mbox{simple-ALSH}\xspace}

\newcommand{\cs}{\mbox{centroid-SLSH}\xspace}
\newcommand{\rs}{\mbox{repeat-SLSH}\xspace}
\newcommand{\es}{\mbox{exhaustive-SLSH}\xspace}
\newcommand{\wes}{\mbox{weighted exhaustive-SLSH}\xspace}
\newcommand{\sls}{\mbox{shrink-lift-SLSH}\xspace}
\newcommand{\SLSH}{\mbox{SLSH}\xspace}
\newcommand{\ALSH}{\mbox{ALSH}\xspace}
\newcommand{\LSH}{\mbox{LSH}\xspace}
\newcommand{\up}[1]{{#1}^{\uparrow}}
\newcommand{\Edst}[2]{d_{\circ}\left(#1,#2\right)}
\newcommand{\EAngd}{d_{\angle\circ}}
\newcommand{\EAngdst}[2]{\EAngd\left(#1,#2\right)}
\renewcommand{\bar}[1]{\mkern 1.5mu\overline{\mkern-1.5mu#1\mkern-1.5mu}\mkern 1.5mu}

\title{Locality Sensitive Hashing for Set-Queries,\protect \newline Motivated by Group Recommendations}
\titlerunning{LSH for set-queries, motivated by group recommendations}

\author{Haim Kaplan}{School of Computer Science, Tel Aviv University, Israel}{haimk@tau.ac.il}{}{}
\author{Jay Tenenbaum}{School of Computer Science, Tel Aviv University, Israel}{jaytenenbaum@mail.tau.ac.il}{}{}

\authorrunning{H. Kaplan and J. Tenenbaum} 

\Copyright{Haim Kaplan, and Jay Tenenbaum} 

\nolinenumbers 
\relatedversion{A full version of the paper is available at \url{https://arxiv.org/abs/2004.07286}.}
\funding{This work was supported by ISF grant no.\ 1595/19 and GIF grant no.\ 1367.}
\hideLIPIcs  

\EventEditors{Susanne Albers}
\EventNoEds{1}
\EventLongTitle{17th Scandinavian Symposium and Workshops on Algorithm Theory (SWAT 2020)}
\EventShortTitle{SWAT 2020}
\EventAcronym{SWAT}
\EventYear{2020}
\EventDate{June 22--24, 2020}
\EventLocation{T\'{o}rshavn, Faroe Islands}
\EventLogo{}
\SeriesVolume{162}
\ArticleNo{27}

\begin{document}
	\maketitle
\begin{abstract}
	Locality Sensitive Hashing (LSH) is an effective method to index a set of points such that we can efficiently find the nearest neighbors of a query point. We extend this method to our novel Set-query LSH (SLSH), such that it can find the nearest neighbors of a set of points, given as a query.
	
	Let $ s(x,y) $ be the similarity between two points $ x $ and $ y $. We define a similarity between a set $ Q$ and a point $ x $ by aggregating the similarities $ s(p,x) $ for all $ p\in Q $. For example, we can take $ s(p,x) $ to be the angular similarity between $ p $ and $ x $ \big(i.e., $1-\frac{\angle (x,p)}{\pi}$\big), and aggregate by arithmetic or geometric averaging, or taking the lowest similarity.
	
	We develop locality sensitive hash families and data structures for a large set of such arithmetic and geometric averaging similarities, and analyze their collision probabilities. We also establish an analogous framework and hash families for distance functions. 
	Specifically, we give a structure for the euclidean distance aggregated by either averaging or taking the maximum.
	
	We leverage \SLSH to solve a geometric extension of the approximate near neighbors problem. In this version, we consider a metric for which the unit ball is an ellipsoid and its orientation is specified with the query.
	
	An important application that motivates our work is group recommendation systems. Such a system embeds movies and users in the same feature space, and the task of recommending a movie for a group to watch together, translates to a set-query $ Q $ using an appropriate similarity.
\end{abstract}
\section{Introduction}
The focus of this paper is on similarity search for queries which are sets of points (set-queries), where we aim to efficiently retrieve points with a high aggregated similarity to the points of the set-query.

Efficient similarity search for massive databases is central in many application areas, such as recommendation systems, content-based image or audio retrieval, machine learning, pattern recognition, and data analysis. The database is often composed of high-dimensional feature vectors of documents, images, etc., and we are interested in finding the near neighbors of a query vector.

Traditional tree-based indexing mechanisms do not scale well to higher dimensions, a phenomenon known as the \enquote{curse of dimensionality}.
To cope with this curse of dimensionality, Indyk and Motwani~\cite{indyk1998approximate,har2012approximate} introduced Locality Sensitive Hashing (LSH), a framework based on hash functions for which the probability of hash collision is higher for similar points than for dissimilar points.

Using such hash functions, one can determine near neighbors by hashing the query point and retrieving the data points stored in its bucket. Typically, multiple LSH functions are concatenated to reduce false positives, and multiple hash tables are needed to reduce false negatives. This gives rise to a data structure which satisfies the following property: for any query point $ q $, if there exists an $ S $-similar data point to $ q $ in the database, it retrieves (with constant probability) some $ cS $-similar data point to $ q $ for some constant $0<c<1$. This data structure is parameterized by a parameter $ \rho =\frac{\log (p_1)}{\log (p_2)}<1 $, where $ p_1 $ is the minimal collision probability for any two points of similarity at least $ S $, and $ p_2 $ is the maximal collision probability for any two points of similarity at most $ cS $.
The data structure can be built in time and space $O(1/p_1\cdot n^{1+\rho})$, and its query time is $ O(1/p_1\cdot n^\rho \log_{1/p_2}(n)) $.

Since the seminal paper of Indyk and Motwani~\cite{indyk1998approximate,har2012approximate}, many extensions have been considered for the LSH framework~\cite{lv2007multi}.
A notable extension is the work of Shrivastava and Li~\cite{shrivastava2014asymmetric}, which study the inner product similarity $ \ipsim(x,y)=x^Ty $. They find near neighbors for the inner product similarity by extending the LSH framework to allow asymmetric hashing schemes (ALSH)~\cite{Neyshabur2014},
in which we hash the query and the data points using different hash functions.
There is also an analogous LSH framework for distance functions, based on hash functions for which the probability of hash collision is higher for near points than for far points.
An important distance function to which the LSH framework has been applied is the $ \ell_p $ distance~\cite{motwani2006lower}. Datar et al.~\cite{datar2004locality} study the $ \ell_p $ distance for $ p\in (0,2]$, and present a hash based on $ p $-stable distributions.
Andoni and Indyk~\cite{andoni2006near} give a near-optimal (data oblivious) scheme for $p=2$. Recently, several theoretically superior data dependent schemes have been designed~\cite{andoni2015optimal,andoni2015tight}.

A noteworthy application of LSH is for \textit{recommendation systems}~\cite{koren2009matrix}, which are required to recommend points that are similar feature-wise to the user. \textit{Group recommendation systems}~\cite{jameson2007recommendation,masthoff2004group} are recommendation systems which provide recommendations, not only to an individual, but also to a whole group of people, and are gaining popularity in recent years. The need in such systems arises in many scenarios: when searching for a movie or a TV show for friends to watch together~\cite{o2001polylens,Yu2006}, a travel destination for a family to spend a holiday break in~\cite{jameson2004more,mccarthy2006needs}, or a good restaurant for a group of tourists to have lunch in~\cite{ardissono2001tailoring}.
In the literature of group recommendation systems, Jameson et al.~\cite{jameson2007recommendation} survey various techniques to aggregate individual user-point similarities $ s $ to a group-point similarity $ s^*$. The most famous aggregation techniques are the \textit{average similarity} which defines the aggregated similarity to be $s^*(Q, x) = \frac{1}{|Q|}\sum_{q\in Q} s(q, x) $, and the \textit{center similarity} (sometimes called \textit{Least-Misery}) which defines the aggregated similarity to be $s^*(Q, x) = \min\limits_{q\in Q}(s(q, x))$.

Most of the work to date on group recommendations is experimental on relatively small data sets. In this paper we give (the first to the best of our knowledge) rigorous mathematical treatment of this problem using the LSH framework.
LSH-based recommendation schemes are used for individual recommendations but do not naturally support group recommendations.
We extend LSH to support set-queries. We formalize this setting by introducing the notions of a set-query-to-point (s2p) similarity function, and of the novel \textit{set-query LSH} (SLSH). 

Our novel set-query LSH (SLSH) framework
extends the LSH framework to similarities between a set of points and a point (s2p similarities).
We define such a similarity between a set-query $ Q=\{q_1,\ldots,q_k\}\subset Z$ and a point
$ x\in Z$ by aggregating (e.g., averaging) point-to-point (p2p) similarities $ (s(q_1,x),\ldots, s(q_k,x)) $ where $ s:Z\times Z\to \reals_{\geq 0} $
is a p2p similarity. 
Specifically, we consider the \textit{$ \ell_p $ similarity} $s_{p}(Q,x)= \frac{1}{k}\sum_{i=1}^{k} \left(s(q_i,x)\right)^p$ for a constant $ p\in \naturals $ (of which the \textit{average similarity} $ s_{avg}(Q,x) =s_{1}(Q,x) $
is a special case), the \textit{geometric similarity} $s_{geo}(Q,x)= \prod_{i=1}^{k} s(q_i,x)$, and the \textit{center similarity} $s_{cen}(Q,x)= \min_{q\in Q} s(q,x) $ of $ s $.\footnote{For ease of presenting our ideas, we define the $ s_p $ and center similarities to be the $ p $'th and $ k $'th power of their conventional definition in the literature. Note that the results follow for the conventional definitions since maximizing a similarity is equivalent to maximizing a constant power of it.} Analogously, we can define s2p distance functions and SLSH framework for distances.
We develop hash families for which the probability of collision between a set-query $Q$ and a point $x$ is higher when $Q$ is similar to $x$ than when $Q$ is dissimilar to $x$.

\paragraph*{Our contribution}
We extend the LSH framework to a novel framework for handling set-queries (SLSH) for both distance and similarity functions, and study their set-query extensions.
We develop various techniques for designing set-query LSH schemes, either by giving an SLSH family directly for the s2p similarity at hand, or by reducing the problem to a previously solved problem for a different distance or similarity.

\subparagraph*{Simple SLSH schemes via achievable p2p similarities.}
We say that a p2p similarity $ s $ is \textit{achievable} if there exists a hash family such that the collision probability between $x$ and $y$ is exactly $ s(x,y) $. The \textit{angular}, \textit{hamming} and \textit{Jaccard} p2p similarities have this property.
We show how to construct SLSH families for the $\ell_p$ and geometric s2p similarities that are obtained by aggregating a p2p similarity which is achievable.

Many of our SLSH families for s2p similarities can be extended to \textit{weighted} s2p similarity functions, in which the contribution of each individual p2p similarity has a different weight.
For example, define the weighted geometric s2p similarity (of a p2p similarity $ s $) of a set-query $ Q $ and a data point $ x $ to be
$ s_{wgeo}(Q,x) = \prod_{i=1}^{k}\left(s\left(q_i,x\right)\right)^{w_i}$.
These weights are independent of the specific query and are given at preprocessing time.
As an example, a solution for the SLSH problem for $ s_{wgeo} $ for any \textit{achievable} p2p similarity $ s $ appears in Appendix~\ref{subsec:weightedgeom}.

Additionally, we present an SLSH scheme for the average euclidean distance which is based upon the shrink-lift transformation (the \enquote{lift} refers to the lifting transformation from Bachrach et al.~\cite{bachrach2014speeding}) which approximately reduces euclidean distances to angular distances. 
We get an average angular distance problem which we then solve using the fact that the angular similarity is achievable and inversely related to the angular distance.\footnote{
	We note that as the LSH approximation parameter $ c $ approaches 1, the required shrink approaches 0. This makes the angles between the lifted points small, which in turn deteriorates the performance of the angular similarity structure (in particular, one can show that the term in $ \log_{1/p_2}(n) $ in the query time bound of the LSH structure approaches infinity). Therefore, we conclude that the shrink-lift transformation is useful for values of $ c $ which are not too close to 1. However, note that such a property holds for any LSH-based nearest neighbors algorithm, where for approximation ratios $ c\to 1 $, the performance becomes equivalent or worse than linear scan.}

\subparagraph*{Ellipsoid ALSH.}
We define the novel \textit{euclidean ellipsoid distance} which naturally extends the regular euclidean distance. We develop an LSH-based near neighbors structure for this distance by a reduction to an SLSH problem with respect to the geometric angular distance. Recall that in the euclidean approximate near neighbor problem, the query specifies the center of two concentric balls such that one is a scaled version of the other. Analogously, in our novel ellipsoid distance, the query specifies the center and orientation of two concentric ellipsoids such that one is a scaled version of the other. If there is a point in the small ellipsoid, we have to return a point in the large one. We reduce this problem to a novel angular ellipsoid distance counterpart via the shrink-lift transformation mentioned before. In this angular distance counterpart, the distance is a weighted sum of squared angles (rather than squared distances in the euclidean ellipsoid distance).

To solve the angular ellipsoid ALSH problem, we make a neat observation that the squared angle that a point creates in the direction of an angular ellipsoid axis, is inversely related to the collision probability of the point with the hyperplane perpendicular to the axis, in the ALSH family of Jain et al.~\cite{jainhashing}. This observation reduces the problem to a weighted geometric angular similarity SLSH problem, which we finally solve as indicated above using the fact that the angular similarity is achievable.

\subparagraph*{Center euclidean distance SLSH.}
The most challenging s2p distance is the center euclidean distance which wants to minimize the maximum distance from the points of the set-query. For this distance function, we obtain an SLSH scheme when the set-query is of size 2, via a reduction to the euclidean ellipsoid ALSH problem. This reduction is based on an observation that the points of center euclidean distance at most $ r $ to a set-query of size 2, approximately form an ellipsoid.

We focus on developing techniques to construct SLSH families, but we do not compute closed formulas for $\rho$ as a function of $S$ and $c$.
These expressions can be easily derived for the simpler families but are more challenging to derive for the more complicated ones.
We leave the optimization of $\rho$ and testing the method on real recommendation data for future work.

\paragraph*{Other related work}
Since we study our novel SLSH framework, there is no direct previous work on this. That been said, there is related previous work on LSH, ALSH, and recommendation systems which are as follows.
In the literature of recommendation systems, Koren and Volinsky~\cite{koren2009matrix} discuss matrix factorization models where user-item interactions are modeled as inner products, and Bachrach et al.~\cite{bachrach2014speeding} propose a transformation that reduces the inner product similarity to euclidean distances. Regarding group recommendation systems, Masthoff and Judith~\cite{masthoff2004group} show that humans care about fairness and avoiding individual misery when giving group recommendations, and Yahia et al.~\cite{amer2009group} formalize semantics that account for item relevance to a group, and disagreements among the group members.
Regarding LSH and ALSH, Neyshabur and Srebro~\cite{Neyshabur2014} study symmetric and asymmetric hashing schemes for the inner product similarity, and show a superior symmetric LSH to that of Shrivastava and Li~\cite{shrivastava2014asymmetric}, that uses the transformation of Bachrach et al.~\cite{bachrach2014speeding}. 
As stated before, we use the ALSH family of Jain et al.~\cite{jainhashing} to solve the angular ellipsoid ALSH problem. We show that this family can be interpreted as a private case of an SLSH family for an appropriate s2p similarity, however Jain et al.~\cite{jainhashing} did not need this property, and the connection is coincidental.

\section{Preliminaries}
We use the following standard definition of a \textit{Locality Sensitive Hash Family (LSH)} with respect to a given point-to-point (p2p) similarity function $ s:Z\times Z\to \reals_{\geq 0} $. 
\begin{definition}[Locality Sensitive Hashing (LSH)]\label{def:LSH}
	Let $ c < 1,~S>0 $ and $ p_1>p_2 $. A family $ H $ of functions $ h:Z\to \Gamma $ is an $ (S, cS, p_1, p_2) $-\LSH for a p2p similarity function $s:Z\times Z \to \reals_{\geq 0}$ if for any $ x,y\in Z$,
	\begin{enumerate}
		\item If $ s(x, y) \geq S $ then $ \Pr_{h\in H}[h(x)=h(y)] \geq p_1 $, and
		\item If $ s(x, y) \leq cS $ then $ \Pr_{h\in H}[h(x)=h(y)] \leq p_2 $.
	\end{enumerate}
\end{definition}

Note that in the definition above, and in all the following definitions, the hash family $ H $ is always sampled uniformly.
Following Shrivastava and Li~\cite{shrivastava2014asymmetric} we extend the LSH framework
to asymmetric similarities $s:Z_1\times Z_2 \to \reals_{\geq 0}$ (where $ Z_1 $ is the domain of the data points and $ Z_2 $ is the domain of the queries). Here the $ (S,cS,p_1,p_2) $-\ALSH family $ H $ consists of pairs of functions $ f:Z_1\to \Gamma $ and  $ g:Z_2\to \Gamma $, and the requirement is that $ \Pr_{(f,g)\in H}[f(x) = g(y)] \geq p_1 $ if $ s(x, y) \geq S $, and $ \Pr_{(f,g)\in H}[f(x) = g(y)] \leq p_2 $ if $ s(x, y) \leq cS $.

\paragraph*{Set-Query LSH}
A special kind of asymmetric similarities are similarities between a set of points and a point (s2p similarities).
That is, similarities of the form  $ s^*:\mathcal{P}(Z,k)\times Z\to \reals_{\geq 0} $, where $ \mathcal{P}(Z,k) $ is the set of subsets of $ Z $ of size $ k $.
We focus on s2p similarity functions that are obtained by aggregating the vector of p2p similarities $ (s(q_1,x),\ldots, s(q_k,x)) $ where $ s:Z\times Z\to \reals_{\geq 0} $ is a p2p similarity function, as we discussed in the introduction. 
We call an $ (S, cS, p_1, p_2) $-\ALSH for an \textbf{s2p} similarity $s^*$, an \textit{ $(S, cS, p_1, p_2)$-\SLSH} for $s^*$.
Our focus is on s2p similarities and \SLSH families.

\paragraph*{From similarities to distances}
For distance functions we wish that close points collide with a higher probability than far points do.
Specifically, we require that $ \Pr_{h\in H}[h(x)=h(y)] \geq p_1 $ if $ d(x, y) \leq r $, that $ \Pr_{h\in H}[h(x)=h(y)] \leq p_2 $ if $ d(x, y) \geq cr $, and that $ c>1 $. We extend the LSH framework for distances to asymmetric distances and for s2p distances, and define ALSH and SLSH families as we did for similarities.
As for similarity functions, we consider s2p distance functions that are defined based on the vector of p2p distances $ (d(q_1,x),\ldots, d(q_k,x)) $. In particular, we consider the \textit{$ \ell_p $ distance} $d_{p}(Q,x)= \frac{1}{k}\sum_{q\in Q} \left(d(q,x)\right)^p$ for a constant $ p\in \naturals $ (of which the \textit{average distance} $ d_{avg}(Q,x) =d_{1}(Q,x) $
is a special case), the \textit{geometric distance } $d_{geo}(Q,x)= \prod_{q\in Q} d(q,x)$, and the \textit{center distance } $d_{cen}(Q,x)= \max_{q\in Q} d(q,x) $ of $ d $, where $ d:Z\times Z \to \reals_{\geq 0} $ is a p2p distance function.

\paragraph*{Additional definitions}
We consider the following common p2p similarity functions $ s : \reals^d \times \reals^d \to \reals_{\geq 0} $: 1) The \textit{angular similarity} $\angle sim(x,y)=1-\frac{\angle(x,y)}{\pi} $, and 2) The \textit{inner product similarity} $\ipsim(x,y)=x^Ty $~\cite{shrivastava2014asymmetric}.
We also consider the following common p2p distance functions $ d : \reals^d \times \reals^d \to \reals_{\geq 0} $: 1) The angular distance $ \angle(x,y) $, and 2) The euclidean distance $ed(x,y)=\norm{x-y}_2$.

We say that a hash family is an \textit{$(S, cS)$-\LSH} for a p2p similarity function $ s $ if there exist $ p_1 > p_2  $ such that it is an $ (S, cS, p_1, p_2) $-\LSH. An $(S, cS)$-\LSH family can be used (see \cite{indyk1998approximate,har2012approximate}) to solve the corresponding $ (S,cS) $-\LSH problem of finding an $ (S,cS) $-\LSH structure. An \textit{$ (S,cS) $-\LSH structure} finds (with constant probability) a neighbor of similarity at least $cS$ to a query $q$ if there is a neighbor of similarity at least $ S $ to $q$. We define these concepts analogously (and apply analogous versions of \cite{indyk1998approximate,har2012approximate}) for ALSH and SLSH hash families and for LSH for distances.

We denote the unit ball in $ \reals^d $ by $ B_d $ and the unit sphere in $ \reals^d $ by $S_d $. We also denote $ [n]: =\{1,\ldots,n\}$, and occasionally use the abbreviations $ (x_i)_{i=1}^{m}:=(x_1,\ldots,x_m) $ and $ \{x_i\}_{i=1}^{m}:=\{x_1,\ldots,x_m\} $. All the missing proofs appear in the appendix.

\section{Similarity schemes}\label{sec:similaritySchemes}
We call a (symmetric or asymmetric) similarity function $ s $ \textit{achievable} if there exists a hash family $ H $ such that for every query $ q $ and point $ x $, $\Pr_{(f,g)\in H}[f(q)=g(x)]= s(q,x)$ (for symmetric p2p similarity functions $ f=g $). Clearly, such an $ H $ is an $ (S,cS) $-\ALSH for $ s $ for any $ S $ and $ c $. In this section, we show that the $ \ell_p $ and geometric s2p similarity functions of an achievable p2p similarity, is by iteself achievable and therefore has and $ (S,cS) $-\SLSH. 

Note that many natural p2p similarity functions are achievable. For example, the random hyperplane hash family~\cite{andoni2006near} achieves the angular similarity function $s(x,y)= 1-\frac{\angle(x,y)}{\pi} $, the random bit hash family~\cite{gionis1999similarity} achieves the hamming similarity $ s\left((x_1,\ldots,x_d),(y_1,\ldots,y_d)\right)=\frac{|\{i|x_i=y_i\}|}{d} $, and MinHash~\cite{broder1997resemblance} achieves the Jaccard similarity $ s(S,T)=\frac{|S\cap T|}{|S\cup T|} $.

In Appendix~\ref{subsubsec:avginnerproduct}, we also give a very simple reduction from the average inner product SLSH problem to the regular inner product ALSH problem (which is not achievable).

\paragraph*{\texorpdfstring{$ \ell_p$}{Lp} similarity}
In this section, we define \textit{\rs}, and prove that it is an SLSH for the $\ell_p$ s2p similarity $ s_p $ of any achievable p2p similarity function $ s $ for any constant $ p\in \naturals $.
The intuition behind \rs is that given an LSH family that achieves a p2p similarity function $ s $, a query point $ q $ collides with a data point $ x $ on $ p $ randomly and independently selected hash functions with probability $ (s(Q,x))^p $. Thus, if we uniformly sample a point $ q\in Q $ of the set-query,\footnote{Therefore, for \rs we do not need to know the set-query size $ k $ a-priori.} and then compute $ p $ consecutive hashes of $ q $, the expected collision probability will be the $ \ell_p $ similarity of $ Q $ and $ x $. The formal definition is as follows.
\begin{definition}[\RS]
	Let $ s $ be an achievable p2p similarity function achieved by a hash family $ H_s $, let $ k $ be the size of the set-query, and let $ p\in \naturals$.
	We define the \textit{\rs of $ H_s $} to be \[ H=\left\{\left(Q \to (h_j(q_i))_{j=1}^p,x\to (h_j(x))_{j=1}^p\right) \mid i\in [k],~(h_1,\ldots,h_p)\in H_s^p\right\},\] where $ q_i $ is the $ i $'th element of the set-query $ Q=\{q_1,\ldots,q_k\} $ in some consistent arbitrary order.\footnote{Let $ A $ be a set, and let $ p\in \naturals $. We define $ A^p:=\{(x_i)_{i=1}^{p}\mid \forall i, x_i\in A\} $.}
\end{definition}

\begin{theorem}\label{thm:prepeat}
	Let $ s $ be an achievable p2p similarity function, and let $ H_s $ be a family that achieves $ s $. Then for any $ S>0 $ and $ c<1 $, the \rs of $ H_s $ is an $ (S,cS) $-SLSH for $ s_p $, the $ \ell_p $ similarity of $ s $.
\end{theorem}
\begin{proof}
	It is clear that $ \Pr_{(f,g)\in H}[f(Q)=g(x)]= s_{p}(Q,x)$ for any set-query $ Q=\{q_i\}_{i=1}^{k} $ and data point $ x $, so it is an $ (S,cS) $-SLSH for any $ S>0 $ and $ c<1 $.
\end{proof}

\paragraph*{Geometric similarity}
The geometric similarity is somewhat similar to the center similarity - both similarities are suitable when we want to enforce high similarity to all points of the set-query.
Analogously, here a query $ Q $ is mapped to $ \left(h_i(q_i)\right)_{i=1}^{k} $ where $ h_1,\ldots,h_k $ are random hash functions, each applied to a corresponding item in $ Q $. A data point $ x $ is mapped to $ \left(h_i(x)\right)_{i=1}^{k} $. It is not hard to see that the collision probability is $ s_{geo}(Q,x) $. In Appendix~\ref{apndx:similaritySchemes}, we give a formal theorem analogous to Theorem~\ref{thm:prepeat} both for the unweighted and weighted versions of the geometric similarity.

\section{Distance schemes}\label{sec:distanceSchemes}
\begin{figure}[t]
	\begin{center}
		\includegraphics[width=0.6\textwidth]{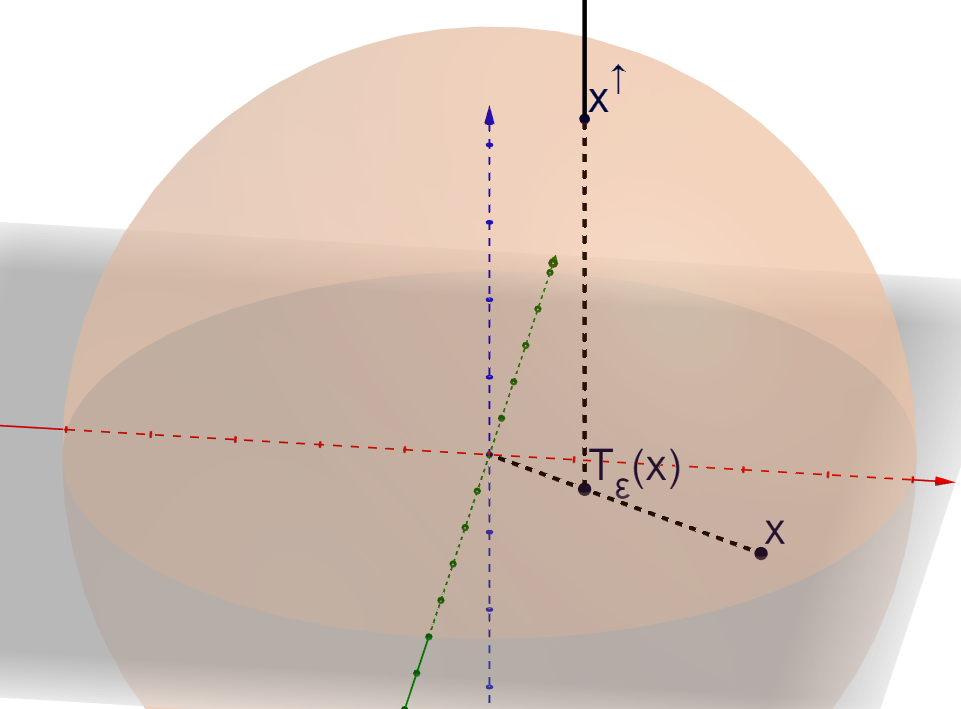}
	\end{center}
	\caption{The shrink-lift transformation $ \up{x} $.}
	\label{fig:shrinklift}
\end{figure}
The notion of achievability that allowed us to construct simple SLSH families for s2p similarity functions does not naturally extend to distance functions. Nevertheless, in this section we directly design two important SLSH families for the average angular and the average euclidean distance functions.

We start with the easy observation that \rs from Section~\ref{sec:similaritySchemes} for $ p=1 $  is, as is, an SLSH family for the average angular distance (the easy proof is in Appendix~\ref{subsec:1avgangulardistance}).\footnote{This family hashes a random point from the set-query $ Q $ to $ \{-1,1\} $ by a random hyperplane.}
In the rest of this section we show how to reduce the average euclidean distance SLSH problem to the average angular distance SLSH problem.
We assume that all data points $ x $ and queries $ Q $ are in $ B_d $, and given the parameters $ r>0$ and $c>1$, we build an $ (r,cr) $-\SLSH structure for the average euclidean distance, $ed_{avg}$, as follows.

We consider the shrink transformation $T_{\eps}:\reals^d \to \reals^d$ defined by $ T_{\eps}(x)=\eps x $ for some $ \eps<\frac{1}{2} $. Additionally, we use the lifting transformation $ L:B_d \to S_{d+1} $ of Bachrach et al.~\cite{bachrach2014speeding},
defined by $ L(x)=\left(x;\sqrt{1-\norm{x}^2}\right) $. For an $ \eps $, which will always be clear from the context, we define the shrink-lift transformation $ \up{(\cdot)}:B_d \to S_{d+1} $, illustrated in Figure~\ref{fig:shrinklift}, by \textit{$ \up{x}:= L(T_{\eps}(x))$}.

The following lemma specifies the relation between the angle of the lifted points and the euclidean distance between the original points. The exact details of the reduction, including the presentation of an \SLSH structure for the average euclidean distance, appear in Appendix~\ref{subsec:1avgeuclideandistance}.
\begin{lemma}
	Let $  x,y \in B_d$ and $ \eps\in (0,\frac{1}{2}] $, and define $ m(x)=\frac{\sqrt{1+2x^2}}{\sqrt{1-2x^2}} $. Then,
	\[\eps\norm{x-y} \leq \angle(\up{x},\up{y})\leq m(\eps)\cdot\eps\norm{x-y}.\]
\end{lemma}

\section{Euclidean ellipsoid ALSH}\label{sec:euclideanEllipsoid}
In this section we present our most technically challenging result — an example that leverages \SLSH to solve a geometric extension of the approximate near neighbor problem for the euclidean distance.
Our structure is built for a specific \enquote{shape} of two concentric ellipsoids (specified by the  weights of their axis), and their \enquote{sizes}, $ r $ and $ cr $, respectively. Given a query which defines the common center and orientation of these ellipsoids, if there is a data point in the smaller $ r $-ellipsoid, then the structure must return a point in the larger $ cr $-ellipsoid. Specifically, we define the euclidean ellipsoid distance as follows.

\paragraph*{Euclidean ellipsoid ALSH}
Let $ q=(p,\{e_i\}_{i=1}^d) $ be a \enquote{query} pair where $ p\in B_d $ is a center of an ellipsoid and $\{e_i\}_{i=1}^d $ are orthogonal unit vectors specifying the directions of the ellipsoid axes, let $ x\in B_d $ be a data point, and let $ \{w_1,\ldots,w_d\} $ be a fixed set of $d$ rational non-negative weights.

We define the \textit{euclidean ellipsoid distance } $ \Edst{q}{x}$ between $q$ and $x$ with respect to the weights
$ \{w_1,\ldots,w_d\} $ to be $ \sum_{i=1}^d w_i \left(e_i^T(x-p)\right)^2$.

In this section, we describe a structure for the euclidean ellipsoid distance $ (r,cr) -$\ALSH problem via a sequence of reductions. We reduce this problem to what we call an \textit{angular ellipsoid \ALSH} problem, which is then solved via another reduction to the \textit{weighted geometric angular similarity \SLSH} problem, which is solved in Appendix~\ref{subsec:weightedgeom}.

\begin{figure}[t]
	\begin{center}
		\includegraphics[width=0.60\textwidth]{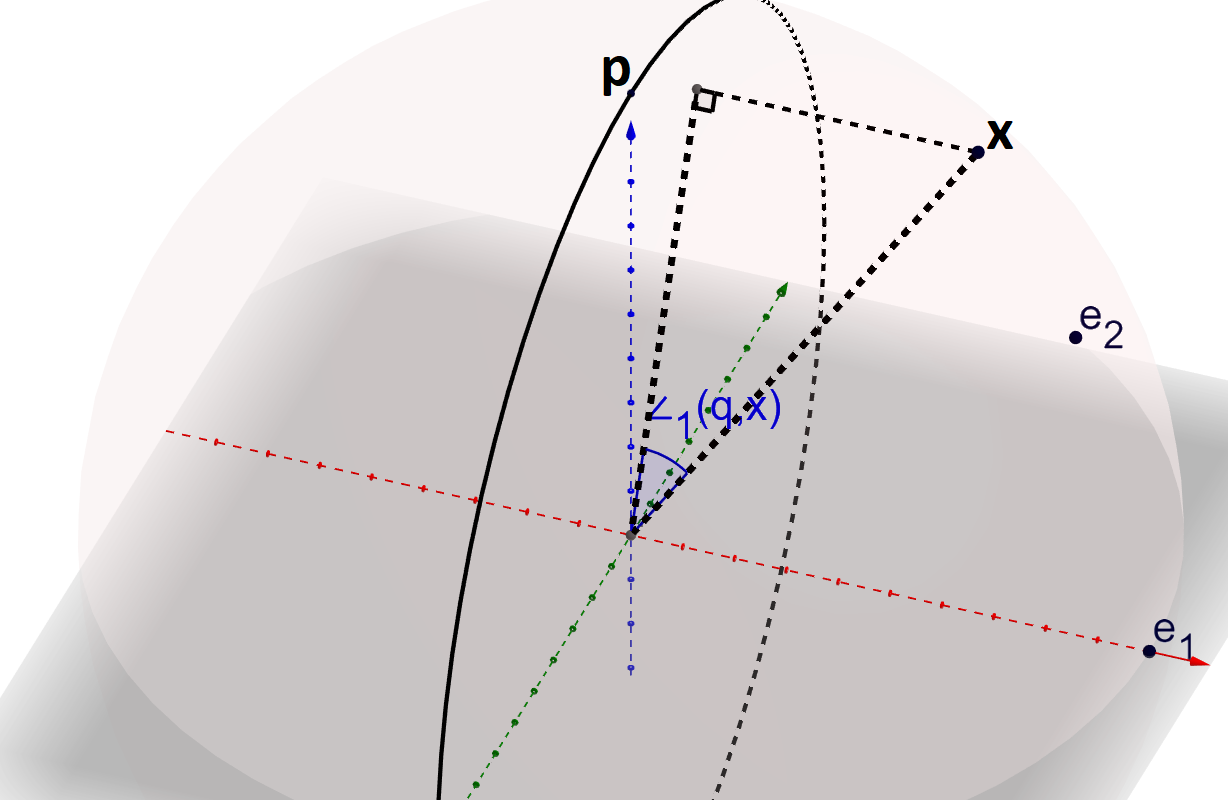}
	\end{center}
	\caption{An angular ellipsoid ALSH query $(p,\{e_i\}_{i=1}^d)  $ and $ \angle_1(q,x) $ for some $ x\in S_{d+1} $.}
	\label{fig:thetai}
\end{figure}

We give a high level description of these reductions and differ the details to Appendix~\ref{apndx:euclideanEllipsoidALSH}. The first reduction is from the euclidean ellipsoid \ALSH to what we call the \textit{angular ellipsoid \ALSH}. Recall that in Section~\ref{sec:distanceSchemes}, we have shown that for small values of $ \eps $, the shrink-lift transformation approximately reduces euclidean distances in $ B_d $ to angular distances on $ S_{d+1} $, for which we can use structures for the angular similarity to solve the associated \SLSH problems.\footnote{As stated in the introduction, we do not want to set $\eps$ to be too small since this deteriorates the performance of subsequent LSH structures we reduce to.} Here, we apply the same shrink-lift transformation to our data, and transform the ellipsoid queries to an angular counterpart defined as follows. 
An angular ellipsoid is specified by a center on the unit sphere and axes perpendicular to it. A point is inside it if the weighted sum of the squared \textbf{angles} that the point creates with the hyperplanes perpendicular to each axis and passing through the origin is smaller than $ r $. 
We formalize this as follows.

\paragraph*{Angular ellipsoid ALSH}
Let $ q=(p,\{e_i\}_{i=1}^d) $ be a \enquote{query} pair where
$ p\in S_{d+1} $ is a center of an \enquote{\textit{angular ellipsoid}}, and $\{e_i\}_{i=1}^d \subset S_{d+1}$ are unit vectors orthogonal to $ p $ (but need not be orthogonal to each other), let $ x\in S_{d+1} $ be a data point, and let $ \{w_1,\ldots,w_d\} $ be a fixed set of $d$
rational non-negative weights.

Given an index $ i\in [d] $, we define $ \angle_i(q,x)\in [0,\frac{\pi}{2}) $ to be the angle between $ x $ and its projection onto the hyperplane through the origin which is orthogonal to $ e_i $. Note that since $ e_i $ is orthogonal to $ p $, this hyperplane contains $ p $. This is illustrated in Figure~\ref{fig:thetai}, 
from which we can also observe that $ \angle_i(q,x)= \sin^{-1}\left(\left| e^T_i \cdot x\right|\right)$.

We define the \textit{angular ellipsoid distance} $ \EAngdst{q}{x}$ between $q$ and $x$ with respect to the weights
$ \{w_1,\ldots,w_d\} $ to be $ \sum_{i=1}^d w_i \cdot \angle_i(q,x)^2$. 

We prove that the shrink-lift transformation approximately maps an ellipsoid to an angular ellipsoid with the same weights, and with a center as the shrink-lift of the original ellipsoid's center, and axes which are slight \enquote{upwards} (to the direction of the axis $ x_{d+1} $) rotations of the axes of the original ellipsoid, such that they are perpendicular to the angular ellipsoid's center (see Figure~\ref{fig:mysphere}). 

\begin{figure}[t]
	\begin{center}
		\includegraphics[width=0.48\textwidth]{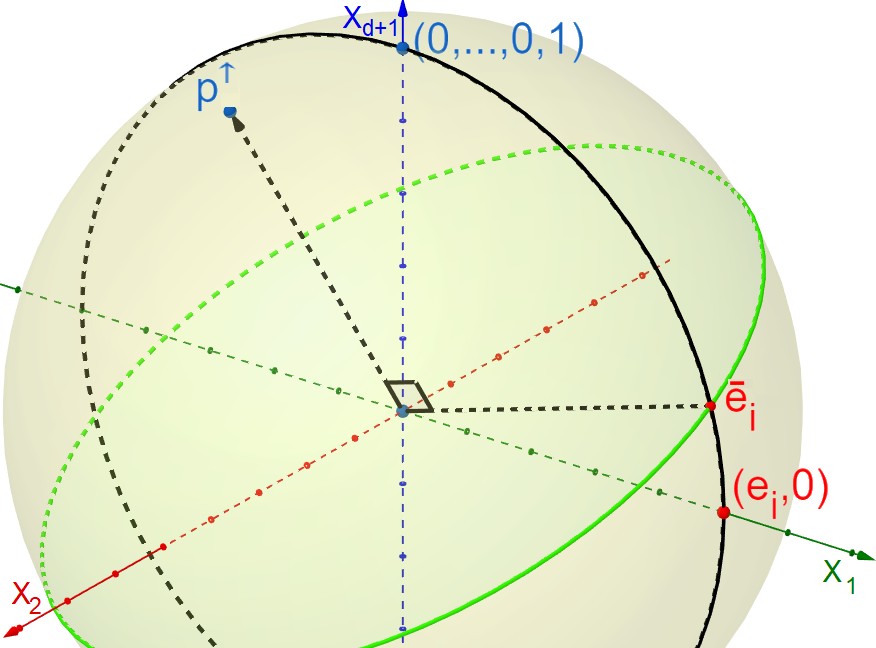}
	\end{center}
	\caption{A query $(p,\{e_i\}_{i=1}^d) $ for the euclidean ellipsoid ALSH, and a corresponding angular axis $ \bar{e_i}$ of $ e_i $.}
	\label{fig:mysphere}
\end{figure}
We solve the angular ellipsoid \ALSH problem by reducing it to the weighted geometric angular similarity \SLSH problem. Our reduction is based on the H-hash of Jain et al.~\cite{jainhashing}, which stores points that reside on $ S_{d+1} $ such that
for a query hyperplane $h$ through the origin, we can efficiently
retrieve the data points that have a small angular distance with their projection on $h$.
H-hash in fact uses an SLSH family for the geometric angular similarity for sets of size $2$, using the following observation which we adapt to our setting. For any direction $ e $ and hyperplane $ h $ perpendicular to $ e $ through the origin, and any $ x\in S_d $, it holds that $\angle sim_{geo}(\left\{e,-e\right\},x)=(1-\angle(x,e)/\pi)(1-\angle(x,-e)/\pi)= \frac{1}{4}-\frac{\angle(x,h)^2}{\pi^2}$, where $ \angle(x,h) $ is the angle between $x$ and its projection on $h$, and the last step follows by the fact that $\min\left(\angle(x,e),\angle(x,-e)\right)= \frac{\pi}{2}- \angle(x,h)$ and $\max\left(\angle(x,e),\angle(x,-e)\right)=\frac{\pi}{2}+\angle(x,h) $.
Recall that the
angular ellipsoid distance between a query $ q=(p,\{e_i\}_{i=1}^d) $
and a point $x$ is a weighted sum of  $ (\angle_i(q,x))^2 $.
Therefore, if we hash the hyperplane orthogonal to $ e_i $ with H-hash, it will collide with higher probability with data points $ x $ with a smaller $ (\angle_i(q,x))^2 $. This suggests that we can answer an angular ellipsoid query $q= (p,\{e_i\}_{i=1}^d) $
by a weighted geometric angular similarity SLSH set-query where the set is the union of the sets $ \{e_i,-e_i\} $ for all $i\in [d] $, using the angular ellipsoid weight $ w_i $ associated with the axis $ e_i $ for each $ i\in [d] $.
Specifically, the corresponding set-query is $Q=\{e_1,-e_1,e_2,-e_2,\ldots,e_d,-e_d\} $, and the structure is built with the weights $\{w_1,w_1,w_2,w_2,\ldots,w_d,w_d\} $.
For the reduction's analysis to hold, we must require that any query $q= (p,\{e_i\}_{i=1}^d) $ and data point $ x $ satisfy $ \angle(p,x) \leq\sqrt{\frac{c-1}{c} }\cdot \frac{\pi}{4} $. This can be easily guaranteed by taking a sufficiently small value of $ \eps $ in the previous reduction from euclidean ellipsoids to angular ellipsoids, such that the set of transformed queries and data points has a sufficiently small angular diameter.

Finally, the weighted geometric angular similarity SLSH problem is solved in Appendix~\ref{subsec:weightedgeom}.

\section{Center euclidean distance for set-queries of size 2}\label{sec:cedforquerysize2}
In this section we present a data structure for the center euclidean $(r,cr)$-\SLSH problem. This is among our most technically challenging results.
Our data structure receives a set-query $ Q=\{q_1,q_2\} $ and returns (with constant probability)
a data point $ v $ such that\\ $ ed_{cen}(Q,v)=\max \left(\norm{v-q_1},\norm{v-q_2}\right)\leq cr $, if there is a data point $ v $ such that\\ $ ed_{cen}(Q,v)=\max\left(\norm{v-q_1},\norm{v-q_2}\right)\leq r $.

Our data structure requires that $c$ is larger than $c_{\min}$ where $c_{\min}=\frac{3}{2\sqrt{2}}\approx 1.06066$ is a constant slightly larger than 1.
We also assume that the possible queries
$ Q=\{q_1,q_2\} $ are such that $\frac{1}{2}\norm{q_1-q_2}<(1-\phi)r$, for a parameter $\phi < 1$ that is known to the structure.\footnote{For queries $ Q=\{q_1,q_2\} $ such that $\frac{1}{2}\norm{q_1-q_2}> r$, no point $ v $ can satisfy $ \max \left(\norm{v-q_1},\norm{v-q_2}\right)\leq r $, and returning no points for such queries satisfies our structure requirements trivially.}

We construct our structure via a reduction to the euclidean ellipsoid ALSH from Section~\ref{sec:euclideanEllipsoid}.

Consider the query $Q=\{q_a,q_{-a}\}$ to the center euclidean SLSH structure where
$q_{a}=(a,0,\ldots,0)$ and $q_{-a}=(-a,0,\ldots,0)$, for some $0<a<(1-\phi)r/2$.
Let $L^s=\{v \mid \max \left(\norm{v-q_a},\norm{v-q_{-a}}\right) \leq r\}$ be the set of point
of center distance at most $r$ from $Q$, and let $L^b=\{v \mid \max \left(\norm{v-q_a},\norm{v-q_{-a}}\right) \leq c r\}$
be the set of point
of center distance at most $cr$ from $Q$.
We also define the following two ellipsoids $S$ and $B$ centered at the origin with axes aligned with the standard axes $x_1,\ldots,x_d$:
\[S=\left\{(x_1,\ldots,x_d) \mid \frac{r+a}{r-a}x_1^2+\sum_{i=2}^{d}x_i^2 \leq r^2-a^2\right\},\]
\[B=\left\{(x_1,\ldots,x_d) \mid \frac{r+a}{r-a}x_1^2+\sum_{i=2}^{d}x_i^2 \leq  \left(\frac{cr}{c_{\min}}\right)^2-a^2\right\}.\]
\begin{figure}[t]
	\begin{center}
		\includegraphics[width=0.35\textwidth]{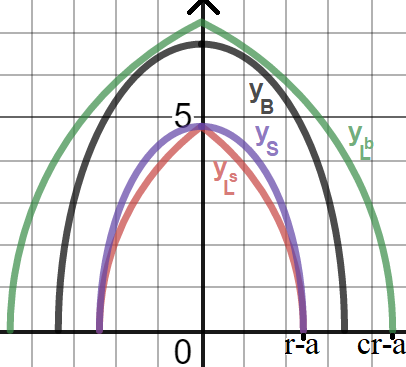}
	\end{center}
	\caption{Plots of $ y_{L^s}$, $y_S$, $y_B$, and $y_{L^b}$ as functions of $ x_1 $. $ a=3.6 $, $ r=6 $, $ c=1.35 $.}
	\label{fig:ellipsoidHierarchy}
\end{figure}

Our reduction depends on the crucial observation stated in the following lemma.

\begin{lemma}\label{lma:ellipsoidsBetweenConstantHeights}
	We have that $ L^s\subseteq S\subseteq B\subseteq L^b $.
\end{lemma}

To illustrate the relation between $L^s$, $S$, $B$, and $L^b$,
we denote the distances of their boundaries from the axis $ x_1 $ by $y_{L^s}(x_1),~y_{S}(x_1),~y_{B}(x_1)$ and $ y_{L^b}(x_1)$, respectively. These functions are plotted in Figure~\ref{fig:ellipsoidHierarchy}.

Intuitively, our reduction will replace $L^s$ by $S$ and $L^b$ by $B$: If there is a point $x$ in $L_s$ then $x$ is also in $S$ and
the euclidean ellipsoid structure will find a point in $B$ which is in $L^b$.
Specifically, we would like to query with $\{q_a,q_{-a}\}$ a euclidean ellipsoid $ (r',c'r') $-\ALSH structure where
$r' = r^2-a^2$, $c'$ is set such that $c'r' =  \left(\frac{cr}{c_{\min}}\right)^2-a^2$, and the weights are $ \left\{\frac{r+a}{r-a},1\ldots,1\right\} $.

The problem is that $a$ depends on the query (it is half the distance between the query points) and
obviously we cannot prepare a different euclidean ellipsoid $ (r',c'r') $-\ALSH structure for each query.
To overcome this we quantize the range of possible values of $a$ and construct a data structure for each
quantized value. The range of the possible values for $a$ is $[0,(1-\phi)r]$ and our quantization consists of the values
$i\cdot \delta$ for $i=0,\ldots \lceil\frac{(1-\phi)r}{\delta}\rceil$ where
$\delta=\min\left(\frac{1}{2},1-\sqrt{\frac{c_{\min}}{c}}\right)\phi r $.\footnote{To ensure rationality of weights, if $ \delta $ is irrational, we replace it by $\rationals_{> 0}\ni \delta'<\delta $.}$ ^, $\footnote{
	Intuitively, when $ c $ is close to $ c_{\min} $, and when $ \phi $ is small, our quantization is finer.}

The euclidean ellipsoid $ (r',c'r') $-\ALSH structure corresponding to the value $i\cdot \delta$ has
$r' =\frac{c}{c_{\min}}\cdot \left(r^2-(i\cdot \delta)^2\right)$, $c'=\frac{c}{c_{\min}}$ and
weights $\left\{\frac{r+i\cdot \delta}{r-i\cdot \delta},1,\ldots,1\right\}$.
For correctness we will prove that the ellipsoids\\ 
$
S^+=\left\{(x_1,\ldots,x_d) \mid \frac{r+a'}{r-a'}x_1^2+\sum_{i=2}^{d}x_i^2 \leq \frac{c}{c_{\min}}\cdot \left(r^2-(a')^2\right)\right\}$ \text{ and}\\
$B^-=\left\{(x_1,\ldots,x_d) \mid \frac{r+a'}{r-a'}x_1^2+\sum_{i=2}^{d}x_i^2 \leq   \left(\frac{c}{c_{\min}}\right)^2\cdot \left(r^2-(a')^2\right)\right\}$,
where $a'=\left\lceil \frac{a}{\delta} \right\rceil \cdot \delta$, are such that $S\subseteq S^+ \subset B^- \subseteq B$. One can easily show that $ r\geq a'\geq 0 $, so the coefficients of $ x_1^2 $ and the right hand side of the equations in $ S^+ $ and $ B^- $ are both non-negative and well-defined.

\paragraph*{Query phase} Let $Q=\{q_1,q_2\}\subseteq B_d $ be a set-query where $\norm{q_1-q_2}=2a $ for $ a\in [0,(1-\phi)r) $.
Let $a'=\left\lceil \frac{a}{\delta} \right\rceil \delta$ as before.
To get the answer, we query the
euclidean ellipsoid $ (r',c'r') $-\ALSH structure,
where $r'= \frac{c}{c_{\min}}\cdot \left(r^2-\left(a'\right)^2\right)$,
$c'= \frac{c}{c_{\min}}$ and the weights are $\left\{\frac{r+a'}{r-a'},1,\ldots,1\right\}$ with a query $q$ defined as follows.

Let $ R_{q_1,q_2} $ be a rigid transformation (rotation and translation) such that $ R_{q_1,q_2}(q_1)=q_a $ and $ R_{q_1,q_2}(q_2)=q_{-a} $ for $ q_a=(a,0\ldots,0)$ and $q_{-a}=(-a,0\ldots,0) $. We set $ q=(p,\{\bar{e_i}\}_{i=1}^d) $ where $p= R_{q_1,q_2}^{-1}\left((0,\ldots,0)\right)=\frac{q_1+q_2}{2}\in B_d $ and $ \forall i,~\bar{e_i}=R_{q_1,q_2}^{-1}\left(e_i\right) $ where $ \{e_i\}_{i=1}^d $ is the standard basis of $ \reals^d $.
Our main result is,
\begin{theorem}\label{thm:1centerToEuclideanEllipsoidReductionSLSH}
	The structure described above is an $ (r,cr) $-\SLSH structure for the center euclidean distance and queries of size 2.
	(For any $ c>c_{\min} $, and queries $ Q=\{q_1,q_2\} $ such that $\frac{1}{2}\norm{q_1-q_2}<(1-\phi)r$.)
\end{theorem}

\section{Conclusions and directions for future work}
We present a novel extended LSH framework, motivated by group recommendation systems.
We define several set-query extensions for distance and similarity functions, and show how to design SLSH families and data structures for them
using different techniques. We use this framework to solve a geometric extension of the euclidean distance approximate near neighbor problem, which we call \textit{euclidean ellipsoid \ALSH}, via reduction to an \SLSH problem.
All the reductions we describe have some performance loss, which (for distance functions) is expressed by a smaller $p_1  $ and $ p_2 $, and a worse value of $ \rho $. Estimating the exact performance loss (the value of $ \rho $) and finding more efficient reductions is an interesting line of research. Finding a method for the center euclidean distance for set-queries larger than two is another intriguing open question.

\section{Acknowledgments}
We want to thank Prof. Micha Sharir and Prof. Edith Cohen for the fruitful discussions.

\bibliography{MDWOM}
\newpage	
\appendix
\section{Missing parts from Section~\ref{sec:similaritySchemes}}\label{apndx:similaritySchemes}
\subsection{Average inner product similarity}\label{subsubsec:avginnerproduct}
The inner product similarity $ \ipsim(q,x)=q^Tx $ is known not to be achievable (see ~\cite{chierichetti2017distortion}), so we cannot use \rs to create an SLSH for $\ipsim_{avg}$ (the average similarity of $ \ipsim$). 
However, we can easily reduce the average inner product SLSH to the inner product ALSH by replacing a set-query $ Q $ by its centroid $ \mu(Q)= \frac{1}{k}\sum_{q\in Q} q $.
For the inner product similarity we can use, for example, the \sa family of Neyshabur and Srebro~\cite{Neyshabur2014}, which is an \textit{ALSH} for $\ipsim$ (Theorem 5.3 in \cite{Neyshabur2014}). Specifically, we define \cs as follows.

\paragraph*{\CS}
We assume that all data points $x$ and set-queries $ Q $ are contained in $ B_d $.
Given the parameters $ S>0,~c<1 $ and the set-query size $ k $, we define the \cs structure to work as follows.
In the preprocessing phase, we store all the data points in an $(S,cS)$-\ALSH structure for $ \ipsim $, and given a set-query $ Q $, we query the $(S,cS)$-\ALSH structure with $\mu(Q) $.

\begin{theorem}
	\CS is an $ (S,cS) $-\SLSH structure for $\ipsim_{avg}$.
\end{theorem}
\begin{proof}
	The claim follows since for every set-query $ Q $ of size $ k $ and data point $ x $,
	\[{\ipsim}_{avg}(Q,x)=\frac{1}{k}\sum_{q\in Q} q^Tx=\mu(Q)^T\cdot x= \ipsim\left(\mu(Q),x\right).\]
\end{proof}

\subsection{Geometric similarity}\label{subsec:1geom}
In this section, we define \textit{\es}, and prove that it is an SLSH for the geometric similarity, $ s_{geo} $, of any achievable p2p similarity function $ s $. 

Note that the geometric similarity is somewhat similar to the center similarity - both similarities are suitable when we want to enforce high similarity to all points of the set-query. Our scheme for center similarity given in Section~\ref{sec:cedforquerysize2} is technically challenging. Thus, \es could be a simple alternative that somewhat relaxes the requirement to be similar to all points of the query for simplicity.

The intuition behind \es is that given an LSH family $ H $ that achieves a p2p similarity function $ s $, then for a set-query $ Q=\{q_1,\ldots,q_k\} $ and a point $ x $, the expected collision probability of $ (h_1(q_1),\ldots,h_k(q_k)) $ with $ (h_1(x),\ldots,h_k(x)) $ when the $ \{h_i\} $'s are sampled from $ H $, is $s_{geo}(Q,x) $. The formal definition is as follows.

\paragraph*{\ES} 
Let $ s $ be an achievable p2p similarity function achieved by a hash family $ H_s $, and let $ k $ be the set-query size.
We define the \es of $ H_s$ to be the following family of pairs\[ H=\left\{\left(Q \to (h_j(q_j))_{j=1}^k,x\to (h_j(x))_{j=1}^k\right) \mid (h_1,\ldots,h_k)\in H_s^k\right\}.\]
\begin{theorem}\label{thm:1geomSensFam}
	Let $ s $ be an achievable p2p similarity, and $ H_s $ be a family that achieves $ s $. Then the \es of $ H_s $ is an SLSH for the geometric similarity of $ s $.
\end{theorem}
\begin{proof}
	Let $ Q =\{q_1,\ldots,q_k\}$ be a set-query of size $ k $. Since $ H_s $ achieves $ s $, for any data point $ x $ we get that
	\begin{align*}
	\Pr_{(f,g)\in H}[f(Q)=g(x)]&=\Pr_{(h_i)_{i=1}^{k}\in H_s^k}\left[\forall j\in [k],h_j(q_j)=h_j(x)\right]\\&=\prod_{j=1}^{k} s(q_j,x)=s_{geo}(Q,x).
	\end{align*}
	Therefore, for any $ S>0 $ and $ c<1 $, the \es of $ H_s $ is an $ (S,cS) $-SLSH for $ s_{geo} $.
\end{proof}

\subsection{Weighted geometric similarity}\label{subsec:weightedgeom}
In this section, we define \textit{\wes}, and prove that it is an \SLSH structure for the weighted geometric similarity $ s_{wgeo} $ of any achievable p2p similarity function $ s $. 
So far, we have only considered equal-weighted query points, however, motivated by recommending movies to a set of people, a logical extension would be giving the individuals weights according to their importance, or the strength of their general preferences. To define the \textit{weighted geometric similarity}, we use a sequence of non-negative \textbf{rational} weights $W= \{w_1,\ldots,w_k\} $, where each $ w_i$ is defined by a pair $ (a_i,b_i) $ such that $a_i\in \naturals\cup\{0\}$, $b_i\in \naturals$, and $w_i =\frac{a_i}{b_i} $, and $ k $ is the set-query size. Given $ W $ and a p2p similarity function $ s $, we define the weighted geometric similarity (of $ s $) of a set-query $ Q=\{q_1,\ldots,q_k\}$ and a data point $ x $ to be
$ s_{wgeo}(Q,x) = \prod_{i=1}^{k}\left(s\left(q_i,x\right)\right)^{w_i}$.\footnote{For weighted similarities we assume that the set-query is ordered, and this order determines the correspondence between the weights and the points in the set-query.}
In case the underlying p2p similarity function $ s $ is achievable, we reduce the weighted geometric similarity $(S,cS)$-\SLSH problem to the geometric similarity $(S',c'S')$-\SLSH problem.

\paragraph*{\WES}
Given $ S>0,~c<1 $, a p2p similarity function $ s $, the set-query size $k$, and non-negative rational weights $ \{w_i\}_{i=1}^{k} $ as defined above, we define $ m=lcm\left(\{b_i\}_{i=1}^{k}\right) \in \naturals$.\footnote{By $ lcm $ we denote the least common multiple.}
The \wes structure works as follows.
In the preprocessing phase, we store all the data points in an $ (S^m,c^mS^m) $-\SLSH structure for the geometric similarity for a set-query of size $k'=m\cdot \sum_{i=1}^{k}w_i$.\footnote{We can derive such a structure from \es (which can be applied since $ s $ is achievable).}
Given a set-query $ Q=\{q_i\}_{i=1}^{k} $, we query the structure built in the preprocessing phase, with the set-query
$ T(Q)=\left\{q_1,\ldots,q_1,\ldots, q_{k},\ldots,q_{k}\right\}$,\footnote{We allow set-queries that are in fact multi-sets. All our derivations apply to multi set-queries.} where each $ q_i\in T(Q)$ is repeated $m\cdot w_i = a_i\cdot\frac{m}{b_i}\in \naturals$ times.

\begin{theorem}\label{thm:qpIsAnSLSH}
	\WES is an $ (S,cS) $-\SLSH structure for the weighted geometric similarity $ s_{wgeo} $ of any achievable p2p similarity function $ s $.
\end{theorem}
\begin{proof}
	Observe that for any set-query $ Q=\{q_i\}_{i=1}^{k} $ of size $ k $ and any data point $ x $, it holds that
	$s_{geo}(T(Q),x)=\prod_{i=1}^{k}\left(s\left(q_i,x\right)\right)^{m\cdot w_i}=\left(\prod_{i=1}^{k}\left(s\left(q_i,x\right)\right)^{w_i}\right)^m=\left(s_{wgeo}(Q,x)\right)^m.$
	Thus, the claim follows since if there is a data point $ x$ such that $ s_{wgeo}(Q,x)\geq S $, then $ s_{geo}(T(Q),x)\geq S^m $, and the $ (S^m,c^mS^m) $-\SLSH structure finds a data point $ x$ such that $s_{geo}(T(Q),x)\geq c^mS^m $, i.e., such that $ s_{wgeo}(Q,x)\geq cS $.
\end{proof}

\section{Detailed results from Section~\ref{sec:distanceSchemes}}
\subsection{Average angular distance}\label{subsec:1avgangulardistance}
We warm up with an easy result, and show that \rs for the average angular \textbf{similarity} (Section~\ref{sec:similaritySchemes}) is an SLSH family for the average angular \textbf{distance} - a fact that follows since the average angular \textbf{similarity} is a decreasing function with respect to the average angular \textbf{distance}.

\begin{theorem}\label{thm:sasIsSLSHforavgangdist}
	\RS for the average angular \textbf{similarity} is an SLSH for the average angular \textbf{distance} $ \angle_{avg} $.
\end{theorem}
\begin{proof}
	Observe that for any set-query $ Q $ of size $ k $ and data point $ x $,
	\begin{align*}
	\angle sim_{avg}(Q,x)=\frac{1}{k}\sum_{q\in Q} \left(1-\frac{\angle(q,x) }{\pi}\right) =1-\frac{\frac{1}{k}\sum_{q\in Q}\angle(q,x)}{\pi}=1-\frac{\angle_{1}(Q,x)}{\pi}.
	\end{align*}
	Thus, the claim follows since for any $ r>0 $ and $ c>1 $, by Theorem~\ref{thm:prepeat}, \rs for the average angular similarity is
	an $ (1-\frac{r}{\pi},1-\frac{cr}{\pi},p_1,p_2) $-\SLSH for $ \angle sim_{1} $ for some $ p_1>p_2 $, and specifically is an $ (r,cr,p_1,p_2) $-\SLSH for $ \angle_{avg} $.
\end{proof}
\subsection{Average euclidean distance}\label{subsec:1avgeuclideandistance}

We give a formal definition of \SL, which reduces the average euclidean distance problem to the average angular distance problem.

\paragraph*{\SL}\label{Subsubsec:shrinkliftslsh}
\SL works as follows.

\subparagraph*{Preprocessing phase.} Given the parameters $ r>0,~c>1 $ and the set-query size $ k $, define $ \eps=\frac{1}{2}\sqrt{1-\frac{2}{1+c^2}}<\frac{1}{2} $. We transform each data point $ x $ to $ \up{x} $, and store the transformed data points in an $ (r',c'r') $-\SLSH structure for average angular distance, for the parameters $r'=m(\eps)\cdot \eps r,~c'=\frac{\eps c r}{r'} = \frac{c}{m(\eps)}$ and $ k'=k$, where we define $ m:\left[0,\frac{1}{2}\right]\to\reals$ by $ m(x)=\frac{\sqrt{1+2x^2}}{\sqrt{1-2x^2}} $.

\subparagraph*{Query phase.} Let $Q$ be a set-query of size $ k $. We query the average angular distance  $ (r',c'r') $-\SLSH structure constructed in the preprocessing phase with the set-query $ Q'=\{\up{q} \mid q\in Q\}$.

In order to prove that \sls is an $(r,cr)$-\SLSH structure for the average euclidean distance, Lemma~\ref{lem:thetaBound} bounds the angle between the lifted points in terms of their original euclidean distance.
It is specified using the error function $e(\eps,x,y):=\left(\sqrt{\frac{1}{\eps^2}-\norm{x}^2}-\sqrt{\frac{1}{\eps^2}-\norm{y}^2}\right)^2 $.

\begin{lemma}\label{lem:thetaBound}
	Let $  x,y \in B_d$ and $ \eps\in (0,1] $. Then
	\[2\sin^{-1}\left(\frac{\eps}{2}\cdot \norm{x-y}\right)\leq \angle(\up{x},\up{y})= 2\sin^{-1}\left(\frac{\eps}{2} \sqrt{\norm{x-y}^2+e(\eps,x,y)}\right).\]
\end{lemma}
\begin{proof}
	
	Let $  x,y \in B_d$. By the definition of $ L(\cdot) $ and the euclidean distance, we get that
	\[ \norm{L(x)-L(y)}=\sqrt{\norm{x-y}^2+\left(\sqrt{1-\norm{x}^2}-\sqrt{1-\norm{y}^2}\right)^2}. \]
	Since $ T_{\eps}(x),T_{\eps}(y)\in B_d $, we can substitute $ x \to T_{\eps}(x), y \to T_{\eps}(y) $ in the equation above. We use the definition of the shrink-lift transformation to conclude that
	\begin{align*}
	\norm{\up{x}-\up{y}}&= \sqrt{\norm{T_{\eps}(x)-T_{\eps}(y)}^2+\left(\sqrt{1-\norm{T_{\eps}(x)}^2}-\sqrt{1-\norm{T_{\eps}(y)}^2}\right)^2}\\
	&=\sqrt{\eps^2 \norm{x-y}^2+\left(\sqrt{1-\eps^2\norm{x}^2}-\sqrt{1-\eps^2\norm{y}^2}\right)^2}\\
	&=\sqrt{\eps^2\norm{x-y}^2+\eps^2\left(\sqrt{\frac{1}{\eps^2}-\norm{x}^2}-\sqrt{\frac{1}{\eps^2}-\norm{y}^2}\right)^2}\\
	&=\eps\sqrt{\norm{x-y}^2+\left(\sqrt{\frac{1}{\eps^2}-\norm{x}^2}-\sqrt{\frac{1}{\eps^2}-\norm{y}^2}\right)^2}\\
	&=\eps\sqrt{\norm{x-y}^2+e(\eps,x,y)},
	\end{align*}
	where the second equality follows since $\norm{T_{\eps}(x)-T_{\eps}(y)}=\eps\cdot \norm{x-y}$ and $ \norm{T_{\eps}(x)}=\eps\norm{x} $.
	
	Thus, we use the fact that the euclidean distance of any two points $ a,b \in S_{d+1}$ is $\norm{a-b}= 2 \sin\left(\frac{\angle(a,b)}{2}\right) $, with the points $ a=\up{x} $ and $ b=\up{y} $, to reason that
	\[\sin\left(\frac{\angle(\up{x},\up{y})}{2}\right)= \frac{1}{2}\norm{\up{x}-\up{y}} =\frac{\eps}{2} \sqrt{\norm{x-y}^2+e(\eps,x,y)}.\]
	Since $ \sin(x/2) $ is increasing for $x\in (0,\pi) $, and $ 0\leq \frac{\angle(\up{x},\up{y})}{2}\leq \frac{\pi}{2} $, we can apply $ \sin^{-1} $ on the equation above and multiply by 2 to get that \[ \angle(\up{x},\up{y})= 2\sin^{-1}\left(\frac{\eps}{2} \sqrt{\norm{x-y}^2+e(\eps,x,y)}\right) \geq 2\sin^{-1}\left(\frac{\eps}{2}\cdot \norm{x-y}\right),\] where the last inequality follows by the non-negativity of $ e(\eps,x,y) $.
\end{proof}

The following lemma bounds the error term.
\begin{lemma}\label{lem:bounde}
	For any $ x,y\in B_d $ and $ \eps\in (0,\frac{1}{2}]$, $0\leq e(\eps,x,y)\leq  \frac{4}{3}\norm{x-y}^2\eps^2$.
\end{lemma}
\begin{proof}
	Let $ x,y\in B_d $ and $ \eps\in (0,\frac{1}{2}] $. The lemma follows by observing that
	\begin{align*}
	e(\eps,x,y)&=\left(\sqrt{\frac{1}{\eps^2}-\norm{x}^2}-\sqrt{\frac{1}{\eps^2}-\norm{y}^2}\right)^2= \left(\frac{\left(\frac{1}{\eps^2}-\norm{x}^2\right)-\left(\frac{1}{\eps^2}-\norm{y}^2\right)}{\sqrt{\frac{1}{\eps^2}-\norm{x}^2}+\sqrt{\frac{1}{\eps^2}-\norm{y}^2}}\right)^2\\
	&= \left(\frac{\left|(\norm{y}+\norm{x})\cdot (\norm{y}-\norm{x})\right|}{\sqrt{\frac{1}{\eps^2}-\norm{x}^2}+\sqrt{\frac{1}{\eps^2}-\norm{y}^2}}\right)^2
	\le \left(\frac{(\norm{y}+\norm{x})\cdot \norm{x-y}}{\sqrt{\frac{1}{\eps^2}-\norm{x}^2}+\sqrt{\frac{1}{\eps^2}-\norm{y}^2}}\right)^2\\
	&\le \left(\frac{2\norm{x-y}}{2\cdot \min \left(\sqrt{\frac{1}{\eps^2}-\norm{x}^2},\sqrt{\frac{1}{\eps^2}-\norm{y}^2}\right)}\right)^2\\
	&= \frac{\norm{x-y}^2}{ \min \left(\frac{1}{\eps^2}-\norm{x}^2,\frac{1}{\eps^2}-\norm{y}^2\right )}\\
	&\le \frac{\norm{x-y}^2}{\frac{1}{\eps^2}-1}=\frac{\norm{x-y}^2}{1-\eps^2}\cdot \eps^2\le \frac{\norm{x-y}^2}{1-\frac{1}{4}}\cdot \eps^2= \frac{4}{3}\norm{x-y}^2\eps^2,
	\end{align*}
	where the first equality follows by the definition of $ e(\eps,x,y) $, the second and third equalities follows from the equation $a-b=\frac{a^2-b^2}{a+b}$, the first inequality follows from the reverse triangle inequality, the second inequality follows since $0\le \norm{x},\norm{y}\le 1$ and $a+b\geq 2\min(a,b)$, and the third and fourth inequalities follow because $ \norm{x},\norm{y}\le 1 $ and $ \eps\leq \frac{1}{2}$.
\end{proof}

Next, we show the following property of $ \sin^{-1}(\cdot) $, which is used in the proof of Lemma~\ref{lma:angleEuclidDistConnection}, and later in the proof of Lemma~\ref{lma:ellipsoidToAngularBound}.
\begin{lemma}\label{lma:invsine}
	$ x\leq \sin^{-1}(x)\leq \frac{x}{\sqrt{1-x^2}} $ for any $x\in[0,1)$.
\end{lemma}
\begin{proof}
	Let $0\leq x< 1$. Since $ \sin^{-1}$ is differentiable in $ [0,x] $, then by Lagrange's mean value theorem, there exists a $ \mu\in (0,x) $ such that $ \sin^{-1}(x)-\sin^{-1}(0)= (\sin^{-1})'(\mu)\cdot (x-0)$, i.e., $ \sin^{-1}(x)= (\sin^{-1})'(\mu)\cdot x$.
	The lemma follows since $(\sin^{-1})'(\mu)= \frac{1}{\sqrt{1-\mu^2}}\in \left(\frac{1}{\sqrt{1-x^2}},1\right)$ for $\mu\in (0,x)$.
\end{proof}

Then, we use Lemmas~\ref{lem:thetaBound},~\ref{lem:bounde} and~\ref{lma:invsine} to derive the following important Lemma.
\begin{lemma}\label{lma:angleEuclidDistConnection}
	Let $  x,y \in B_d$ and $ \eps\in (0,\frac{1}{2}] $. Then,
	$\eps\norm{x-y} \leq \angle(\up{x},\up{y})\leq m(\eps)\cdot\eps\norm{x-y}.$
\end{lemma}
\begin{proof}
	By Lemma~\ref{lem:bounde} and Lemma~\ref{lem:thetaBound}, we deduce that
	\begin{equation}\label{eqn:thetaBound}
	2\sin^{-1}\left(\frac{\eps}{2}\cdot \norm{x-y}\right)\leq \angle(\up{x},\up{y})\leq 2\sin^{-1}\left(\frac{\eps}{2} \sqrt{1+\frac{4\eps^2}{3}}\cdot \norm{x-y}\right).
	\end{equation}
	Recall that $ \eps\leq\frac{1}{2} $ and $ \norm{x-y}\leq 2 $, so the arguments $ \frac{\eps}{2}\cdot \norm{x-y}$ and $ \frac{\eps}{2} \sqrt{1+\frac{4\eps^2}{3}}\cdot \norm{x-y} $ from Inequality~(\ref{eqn:thetaBound}) are both in $[0,1) $. Thus, we use that $ \eps\leq\frac{1}{2} $ and $ \norm{x-y}\leq 2 $ together with Lemma~\ref{lma:invsine}, to deduce that
	\begin{align*}
	\eps\norm{x-y}&
	\leq \angle(\up{x},\up{y})\leq
	2\cdot \frac{\frac{\eps}{2} \sqrt{1+\frac{4\eps^2}{3}}\cdot \norm{x-y}}{\sqrt{1-\left(\eps\sqrt{1+\frac{4\eps^2}{3}}\right)^2}}\\
	&= \frac{\sqrt{1+\frac{4\eps^2}{3}}}{\sqrt{1-\eps^2\left(1+\frac{4\eps^2}{3}\right)}}\cdot\eps\norm{x-y}
	\leq
	\frac{\sqrt{1+2\eps^2}}{\sqrt{1-\eps^2\cdot 2}}\cdot\eps\norm{x-y}\\
	&=
	m(\eps)\cdot\eps\norm{x-y},
	\end{align*}
	where the last inequality follows since $1+ \frac{4\eps^2}{3}\leq  1+2\eps^2\leq 1+2\cdot \frac{1}{4}<2 $.
\end{proof}

Finally, we use Lemma~\ref{lma:angleEuclidDistConnection} to prove the following theorem, which is the main result of this section.

\begin{theorem}\label{thm:slisslshforavgeucldist}
	\SL is an $ (r,cr) $-\SLSH structure for the average euclidean distance $ ed_{avg} $.
\end{theorem}
\begin{proof}
	Consider a set-query $ Q$ of size $ k $ for the average euclidean $ (r,cr) $-\SLSH structure, and let $ Q'=\{\up{q} \mid q\in Q\} $ be the corresponding query for the average angular distance $(r',c'r')$-\SLSH structure. It suffices to prove that: 
	\begin{enumerate}
		\item $ c'>1 $,
		\item $\forall x~s.t.~ed_{avg}(Q,x)\leq r,~\angle_{avg}(\up{Q},\up{x})\leq r',$ and
		\item $\forall x~s.t.~ed_{avg}(Q,x)> cr,~\angle_{avg}(\up{Q},\up{x})> c'r'$.
	\end{enumerate}
	The proofs of these facts are as follows.
	\begin{enumerate}
		\item Observe that $ m(\eps)=\sqrt{\frac{1+2\cdot\frac{1}{4}\cdot \left(1-\frac{2}{1+c^2}\right)}{1-2\cdot\frac{1}{4}\cdot \left(1-\frac{2}{1+c^2}\right)}}=\sqrt{\frac{1\frac{1}{2}-\frac{1}{1+c^2}}{\frac{1}{2}+\frac{1}{1+c^2}}}=\sqrt{\frac{3(1+c^2)-2}{1+c^2+2}}=\sqrt{\frac{1+3c^2}{3+c^2}}< \sqrt{\frac{4c^2}{4}}=c$, where the first equality follows by the definition of $ m(\cdot) $ and since we have taken $ \eps=\frac{1}{2}\sqrt{1-\frac{2}{1+c^2}} $, and the inequality follows since $ c>1 $. Thus, $ c'= \frac{c}{m(\eps)}>1$.
		\item Assume that $ed_{avg}(Q,x)\leq r$. We prove that $\angle_{avg}(\up{Q},\up{x})\leq r'$. Indeed, by Lemma~\ref{lma:angleEuclidDistConnection}
		\begin{align*}
		\angle_{avg}(\up{Q},\up{x})&=\frac{1}{k}\sum_{q\in Q} \angle(\up{q},\up{x}) \leq
		\frac{1}{k}\sum_{q \in Q} m(\eps)\cdot \eps \norm{q-x}\\
		&= m(\eps)\cdot \eps \cdot ed_{avg}(Q,x)
		\leq m(\eps)\cdot \eps r = r'.
		\end{align*}
		\item Assume that $ed_{avg}(Q,x)>cr$. We prove that $ \angle_{avg}(\up{Q},\up{x})> c'r' $. Indeed, by Lemma~\ref{lma:angleEuclidDistConnection}
		\begin{align*}
		\angle_{avg}(\up{Q},\up{x})&=\frac{1}{k}\sum_{q \in Q} \angle(\up{q},\up{x}) \geq  \frac{1}{k}\sum_{q \in Q} \eps\norm{q-x}\\
		&=\eps\cdot ed_{avg}(Q,x)>
		\eps cr = c'r'.
		\end{align*}
	\end{enumerate}
\end{proof}

\section{Euclidean ellipsoid ALSH detailed presentation}\label{apndx:euclideanEllipsoidALSH}

In this section, we give a detailed presentation of the two reductions we use to solve the euclidean ellipsoid ALSH problem from Section~\ref{sec:euclideanEllipsoid}. Section~\ref{subsec:eucelltoangell} gives a reduction from the euclidean ellipsoid ALSH to the angular ellipsoid problem. Section~\ref{subsec:angellto1geomangsim} then reduces this problem to 
the weighted geometric angular similarity \SLSH problem, which is solved in Appendix~\ref{subsec:weightedgeom}. 
We note that this reduction requires that any query $q=(p,\{e_i\}_{i=1}^d)$ and data point $ x$ in the
angular ellipsoid structure
satisfy $ \angle(p,x)\leq \sqrt{\frac{c-1}{c} }\cdot \frac{\pi}{4}$.
As we will see, the inputs to the angular ellipsoid structure that we produce by the first reduction (i.e., from
the euclidean ellipsoid problem) will satisfy this requirement.

It is worth mentioning that the solution in Appendix~\ref{subsec:weightedgeom} requires that the weights are rational, hence we also require rational weights in both the ellipsoid structures.

\subsection{From euclidean ellipsoid ALSH to angular ellipsoid ALSH}\label{subsec:eucelltoangell}
In this section, we reduce the euclidean ellipsoid $(r,cr)$-\ALSH problem to an angular ellipsoid $(r',c'r')$-\ALSH problem.
To do this, we use the shrink-lift transformation $ \up{(\cdot)} $ from Section~\ref{sec:distanceSchemes}
with an appropriately tuned shrinking parameter $ \eps $, to map our data points from $B_d$ to $S_{d+1}$.
For our proofs of Lemma~\ref{lma:ellipsoidToAngularBound} and Theorem~\ref{thm:ellipsoidToAngular} to hold, we need that $ \eps\leq \frac{1}{8} $. Additionally, to prove that the parameter $ c' $ that we use for the angular ellipsoid $ (r',c'r') $ structure is larger than 1 (Theorem~\ref{thm:ellipsoidToAngular}), we need that $ \eps\leq \frac{\sqrt[8]{c}-1}{\sqrt[8]{c}+1} $ and $\eps\leq  \sqrt{\frac{(c-\sqrt{c})r}{5(\sqrt{c}+1)\cdot\sum_{i=1}^d w_i}} $. Finally, to ensure that $ \angle(p,x)\leq \sqrt{\frac{c-1}{c} }\cdot \frac{\pi}{4}$ for any query $q=(p,\{e_i\}_{i=1}^d)$ and data point $ x$ in the angular ellipsoid structure (see the proof of Theorem~\ref{thm:ellipsoidToAngular}), we need that $\eps\leq \sqrt{1-\frac{1}{\sqrt[4]{c}} }\cdot \frac{\pi}{8\sqrt{2}}$.
We therefore set $ \eps $ to be the minimum of all these upper bounds, that is $ \eps=\min\left(\frac{1}{8},\frac{\sqrt[8]{c}-1}{\sqrt[8]{c}+1}, \sqrt{\frac{(c-\sqrt{c})r}{5(\sqrt{c}+1)\cdot\sum_{i=1}^d w_i}},
\sqrt{1-\frac{1}{\sqrt[4]{c}} }\cdot \frac{\pi}{8\sqrt{2}}\right).$

We store the images (by the shrink-lift transformation) of our data points in the angular ellipsoid $(r',c'r')$-\ALSH structure.\footnote{We do not want to set $\eps$ to be too small since this is likely to deteriorate the performance of
	the angular ellipsoid structure on these images.}
We recall (Lemma~\ref{lma:angleEuclidDistConnection}) that for a sufficiently small $ \eps $
the angular distance between $\up{x}$ and $\up{y}$ is approximately equal to $\eps$ times the euclidean
distance between $x$ and $y$.
We set $
r' = \eps^2(1+\eps)^2\cdot \left(r+5\beta(\eps)\cdot\sum_{i=1}^d w_i\right), \text{ and }
c'=\frac{\eps^2(1-\eps)^2\cdot\left(cr-5\beta(\eps)\cdot \sum_{i=1}^d w_i\right)}{r'},
$
where $ \beta(\eps) =\frac{1-\sqrt{1-\eps^2}}{\sqrt{1-\eps^2}}\approx \frac{\eps^2}{2} \geq0$.
Our choice of $\eps$  guarantees that $ \beta (\eps)\cdot \sum_{i=1}^{d}w_i\ll r $
and thereby $ r' $ is approximately $ \eps^2\cdot r $, as we expect since the angular ellipsoid distance is
a sum of (weighted) squared angular distances each of which is smaller by a factor of $\eps$ from its corresponding euclidean
distance. Notice also that for our choice of $\eps$, $c'$ is approximately equal to $\sqrt[4]{c}$.\footnote{
	By using a smaller $\eps$ we can make $c'$ closer to $c$.}
The angular ellipsoid structure uses the same weights as of the euclidean ellipsoid structure.

\paragraph*{The query} Let $q_0= (p,\{e_i\}_{i=1}^d) $ be a euclidean ellipsoid query, where 
$ p\in B_d $ is a center of an ellipsoid and $\{e_i\}_{i=1}^d $ are the unit vectors of $ \reals^d $ in the directions of the ellipsoid axes.
We query the angular ellipsoid structure constructed in the preprocessing phase with the angular ellipsoid query $ q=(\up{p},\{\bar{e_i}\}_{i=1}^d) $, where each $ \bar{e_i} $ is obtained by rotating $ (e_i,0) $ in the direction of $(0,\ldots,0,1) $, until its angle with $ \up{p} $ becomes $ \frac{\pi}{2} $ (this is illustrated in Figure~\ref{fig:mysphere}).
Formally, we define $ \bar{e_i}:= (a_i\cdot e_i;\sqrt{1-a_i^2}) $ where $ a_i=-sign(p_i)\cdot \sqrt{\frac{1-\norm{\eps p}^2}{\eps^2 p_i^2 +1-\norm{\eps p}^2}}\in [-1,1] $, and
$ sign(x)=\begin{cases}
1  & \text{if } x \ge 0 \\
-1  & \text{if } x < 0
\end{cases}  $.
To simplify the expression above, we define $z(p,\eps):=  \sqrt{\eps^2 p_i^2 +1-\norm{\eps p}^2}$, so we get that
\begin{equation}\label{eq:adef}
a_i= -sign(p_i)\cdot\frac{\sqrt{1-\norm{\eps p}^2}}{z(p,\eps)},~ and~\sqrt{1-a_i^2}=\frac{ \eps |p_i|}{z(p,\eps)}=\frac{ \eps p_i\cdot sign(p_i)}{z(p,\eps)}.
\end{equation}
Note that this definition of $ \bar{e_i} $ makes $ \bar{e_i} $ orthogonal to $ \up{p} $. Indeed, 
\begin{align*}
\bar{e_i}^T\cdot  \up{p} &=a_i\cdot \eps p_i +\sqrt{1-a_i^2}\cdot \sqrt{1-\norm{\eps p}^2} \\&=-sign(p_i)\cdot\frac{\sqrt{1-\norm{\eps p}^2}}{z(p,\eps)} \cdot \eps p_i+\frac{ \eps p_i\cdot sign(p_i)}{z(p,\eps)}\cdot \sqrt{1-\norm{\eps p}^2}=0,
\end{align*}
where the first equality follows from the definition $ \up{x} =(\eps x_1,\ldots,\eps x_d,\sqrt{1-\norm{\eps x}^2})$.

The following Lemma implies the correctness of our reduction and the resulting structure, stated in Theorem \ref{thm:ellipsoidToAngular}.

\begin{lemma}\label{lma:ellipsoidToAngularBound}
	Let $ \eps\in (0,\frac{1}{2}),~ x,p\in B_d,$ and a euclidean ellipsoid query $q_0= (p,\{e_i\}_{i=1}^d) $, where $\{e_i\}_{i=1}^d$ is the standard basis in $ \reals^d $. Then for $  q=(\up{p},\{\bar{e_i}\}_{i=1}^d) $ as defined above we have that for every $i\in [d] $
	\[ \max \left(0,\eps(1-\eps)\cdot (|x_i-p_i|-\beta(\eps))\right)\leq \angle_i(q,\up{x}) \leq \eps(1+\eps)\cdot(|x_i-p_i|+\beta(\eps)),\]
	where
	$ \angle_i(q,x) $  is the angular distance between  $ x $ and its projection on the hyperplane orthogonal to $ e_i $ (Figure~\ref{fig:thetai}).
\end{lemma}
\begin{proof}
	Since $ \up{x}=\left(\eps x;\sqrt{1-\norm{\eps x}^2}\right)\in S_{d+1} $, we get that
	\begin{align}
	\angle_i(q,\up{x})&=
	\sin^{-1}\left(\left| \bar{e_i}^T \cdot  \up{x}\right|\right)
	= \sin^{-1}\left(\left| a_i\cdot e_i^T\cdot \eps x+\sqrt{1-a_i^2}\cdot \sqrt{1-\norm{\eps x}^2} \right|\right)\nonumber\\
	&=\sin^{-1}\left(\left|-sign(p_i) \frac{\sqrt{1-\norm{\eps p}^2}}{z(p,\eps)}\cdot \eps x_i+\frac{ \eps p_i\cdot sign(p_i)}{z(p,\eps)}\cdot \sqrt{1-\norm{\eps x}^2} \right|\right)\nonumber\\
	&=\sin^{-1}\left(\frac{\eps}{z(p,\eps)}\left|-\sqrt{1-\norm{\eps p}^2}\cdot x_i+ p_i\cdot \sqrt{1-\norm{\eps x}^2} \right|\right)\nonumber\\
	&=\sin^{-1}\left(\frac{\eps\sqrt{1-\norm{\eps p}^2}}{z(p,\eps)}\left|x_i-\frac{\sqrt{1-\norm{\eps x}^2}}{\sqrt{1-\norm{\eps p}^2}}p_i  \right|\right)\nonumber\\
	&=\sin^{-1}\left(\frac{\eps\sqrt{1-\norm{\eps p}^2}}{z(p,\eps)}\left|x_i-p_i+\frac{\left(\sqrt{1- \norm{\eps p}^2}-\sqrt{1- \norm{\eps x}^2}\right)}{\sqrt{1-\norm{\eps p}^2}} p_i\right|\right)\label{eq:angleiqx},
	\end{align}
	where the second equality follows from the definitions of $ \bar{e_i} $ and $ \up{x} =(\eps x_1,\ldots,\eps x_d,\\\sqrt{1-\norm{\eps x}^2})$, the third equality follows from Equation~(\ref{eq:adef}), the fourth equality follows by the fact that $ |sign(\cdot)|=1 $, and the two last equalities follow since $|x|=|-x|$ and by adding and subtracting $ p_i $.
	
	\paragraph*{Right Inequality}
	Observe that the following holds,
	\begin{align}
	\left|x_i-p_i+\frac{\sqrt{1- \norm{\eps p}^2}-\sqrt{1- \norm{\eps x}^2}}{\sqrt{1-\norm{\eps p}^2}} p_i\right|&\leq
	|x_i-p_i|\nonumber\\&+\frac{\left|\sqrt{1- \norm{\eps p}^2}-\sqrt{1- \norm{\eps x}^2}\right|}{\sqrt{1-\norm{\eps p}^2}}\cdot |p_i|\nonumber\\
	&\leq|x_i-p_i|+ \frac{1-\sqrt{1-\eps^2}}{\sqrt{1-\eps^2}}\nonumber\\
	&= |x_i-p_i|+\beta(\eps)\label{eq:rightineqFirst},
	\end{align}
	where the first inequality follows from the triangle inequality, and the last inequality follows since $ \norm{\eps p}^2,\norm{\eps x}^2 \in \left[0 ,\eps^2\right]$ and $ |p_i| \leq 1$.
	Thus, by substituting Equation~(\ref{eq:rightineqFirst}) into Equation~(\ref{eq:angleiqx}) we get that	
	\begin{align*}
	\angle_i(q,\up{x})&\leq \sin^{-1}\left(\frac{\eps\sqrt{1-\norm{\eps p}^2}}{z(p,\eps)}\cdot (|x_i-p_i|+\beta(\eps))\right)\\
	&\leq \frac{|x_i-p_i|+\beta(\eps)}{\sqrt{1-9\eps^2}}\cdot\frac{\eps\sqrt{1-\norm{\eps p}^2}}{z(p,\eps)}\\
	&  \leq \frac{1}{\sqrt{1-9\eps^2}}\cdot\frac{\eps\sqrt{1-\norm{\eps p}^2}}{\sqrt{1-\norm{\eps p}^2}}\cdot (|x_i-p_i|+\beta(\eps))\\
	&=
	\frac{\eps}{\sqrt{1-9\eps^2}}\cdot (|x_i-p_i|+\beta(\eps))\\
	&\leq \eps(1+\eps)\cdot (|x_i-p_i|+\beta(\eps)),
	\end{align*}
	where the first inequality follows by Equation~(\ref{eq:angleiqx}) and since $ \sin^{-1}(x) $ is an increasing function for $x\in(-\frac{\pi}{2},\frac{\pi}{2})  $, the second inequality follows by using Lemma~\ref{lma:invsine} together with the fact that the argument of the $ \sin^{-1}(\cdot) $ is at most $ 3\eps$.\footnote{The argument of the $ \sin^{-1}(\cdot) $ is at most $ \frac{3\eps}{\sqrt{1-\eps^2}}$ since $ \beta(\eps)<1 $ for $ \eps\leq\frac{1}{8} $, and $ |x_i-p_i|\leq 2$ and $ z(p,\eps)=\sqrt{\eps^2 p_i^2 +1-\norm{\eps p}^2}\geq \sqrt{1-\norm{\eps p}^2}$.} The third inequality follows since $ z(p,\eps)\geq \sqrt{1-\norm{\eps p}^2} $, and the last inequality follows since $ \eps\leq\frac{1}{8} $.
	
	\paragraph*{Left Inequality}
	As in the proof of the right inequality, we get that
	\begin{align*}
	\left|x_i-p_i+\frac{\sqrt{1- \norm{\eps p}^2}-\sqrt{1- \norm{\eps x}^2}}{\sqrt{1-\norm{\eps p}^2}} p_i\right|&\geq
	|x_i-p_i|\\
	&-\frac{\left|\sqrt{1- \norm{\eps p}^2}-\sqrt{1- \norm{\eps x}^2}\right|}{\sqrt{1-\norm{\eps p}^2}}\cdot |p_i|\\
	&\geq|x_i-p_i|- \frac{1-\sqrt{1-\eps^2}}{\sqrt{1-\eps^2}}\\
	&=|x_i-p_i|-\beta(\eps).
	\end{align*}
	Thus, since $ \sin^{-1}(x) $ is an increasing function for $x\in(-\frac{\pi}{2},\frac{\pi}{2})  $, we conclude that the following holds, $ \angle_i(q,\up{x})\geq \sin^{-1}\left(\frac{\eps\sqrt{1-\norm{\eps p}^2}}{z(p,\eps)}\cdot (|x_i-p_i|-\beta(\eps))\right) $.
	If $ |x_i-p_i|< \beta(\eps) $, then the left inequality holds since we always have that $ \angle_i(q,\up{x})\geq 0$. Otherwise, assume that $ |x_i-p_i|\geq \beta(\eps) $, so
	\begin{align*}
	\angle_i(q,\up{x})&\geq \sin^{-1}\left(\frac{\eps\sqrt{1-\norm{\eps p}^2}}{z(p,\eps)}\cdot (|x_i-p_i|-\beta(\eps))\right)\\
	&\geq \frac{\eps\sqrt{1-\norm{\eps p}^2}}{z(p,\eps)}\cdot (|x_i-p_i|-\beta(\eps))\\
	&\geq \eps\sqrt{1-\eps^2}\cdot (|x_i-p_i|-\beta(\eps))
	\geq \eps(1-\eps)\cdot (|x_i-p_i|-\beta(\eps)),
	\end{align*}
	where the second inequality follows by Lemma~\ref{lma:invsine}, which we can apply since the argument of $ \sin^{-1}(\cdot) $ is in $ [0,1) $, and the third inequality follows since $ z(p,\eps)\leq 1 $.
\end{proof}

In Section~\ref{subsec:angellto1geomangsim}, we show the existence of an angular ellipsoid $(r',c'r')$-\ALSH structure, so we conclude the following theorem.
\begin{theorem}\label{thm:ellipsoidToAngular}
	The structure above is an $ (r,cr) $-\ALSH structure for the euclidean ellipsoid distance $ d_{\circ} $.
\end{theorem}
\begin{proof}
	Consider a query $ q_0= (p, \{e_i\}_{i=1}^d)$ for the euclidean ellipsoid ALSH structure, and let $ q=(\up{p},\{\bar{e_i}\}_{i=1}^d)$ be the corresponding query for the angular ellipsoid $(r',c'r')$-\ALSH structure. We assume w.l.o.g. that $\{e_i\}_{i=1}^d$ is the standard basis of $ \reals^d $.
	It suffices to prove that
	\begin{enumerate}
		\item $c'>1,$
		\item $\forall p,x\in B_d,~ \angle(\up{p},\up{x})\leq \sqrt{\frac{c'-1}{c'} }\cdot \frac{\pi}{4}, $
		\item $\forall x~s.t.~\Edst{q_0}{x}\leq r,~\EAngdst{q}{\up{x}}\leq r',$ and
		\item $\forall x~s.t.~\Edst{q_0}{x}> cr,~\EAngdst{q}{\up{x}}> c'r'$.
		
	\end{enumerate}
	The proofs of these claims are as follows.
	\begin{enumerate}
		\item Observe that since $ \eps\in (0,\frac{1}{2}) $, we have that $ \beta(\eps)=\frac{1-\sqrt{1-\eps^2}}{\sqrt{1-\eps^2}}=\frac{1}{\sqrt{1-\eps^2}}-1<1+\eps^2-1=\eps^2$, so using the definitions of $ c' $ and $ r' $ we get that
		\begin{equation}
		c'=\frac{(1-\eps)^2}{(1+\eps)^2}\cdot \frac{cr-5\beta(\eps)\cdot \sum_{i=1}^d w_i}{r+5\beta(\eps)\cdot\sum_{i=1}^d w_i}
		>\left(\frac{1-\eps}{1+\eps}\right)^2\cdot \frac{cr-5\eps^2\cdot \sum_{i=1}^d w_i}{r+5\eps^2\cdot\sum_{i=1}^d w_i}\geq \frac{1}{\sqrt[4]{c}}\cdot \sqrt{c}=\sqrt[4]{c},\label{ineq:ctaglowerbound}
		\end{equation}
		where the second inequality follows since our choice of $ \eps\leq \frac{\sqrt[8]{c}-1}{\sqrt[8]{c}+1} $ implies that $ \left(\frac{1-\eps}{1+\eps}\right)^2\geq \frac{1}{\sqrt[4]{c}} $, and since our choice of $\eps\leq  \sqrt{\frac{(c-\sqrt{c})r}{5(\sqrt{c}+1)\cdot\sum_{i=1}^d w_i}} $ implies that $  \frac{cr-5\eps^2\cdot \sum_{i=1}^d w_i}{r+5\eps^2\cdot\sum_{i=1}^d w_i}\geq \sqrt{c} $. Since $ c>1 $, Inequality~(\ref{ineq:ctaglowerbound}) implies claim 1.
		
		\item Let $ p,x\in B_d $, and recall from the \SL paragraph that $ m(x)=\frac{\sqrt{1+2x^2}}{\sqrt{1-2x^2}}= \sqrt{1+\frac{4\eps^2}{1-2\eps^2}}$. By Lemma~\ref{lma:angleEuclidDistConnection}, we have that
		\begin{align*}
		\angle(\up{p},\up{x})&\leq m(\eps)\cdot\eps\norm{p-x}\leq
		m(\eps)\cdot 2\eps
		=2\eps\sqrt{1+\frac{4\eps^2}{1-2\eps^2}}\\
		&< 2\eps\sqrt{1+\frac{4\cdot {\frac{1}{4}}^2}{1-2\cdot {\frac{1}{4}}^2}}
		=2\eps\sqrt{1+\frac{8}{4\cdot 7}}\leq 2\sqrt{2}\eps\leq
		\sqrt{1-\frac{1}{\sqrt[4]{c}} }\cdot \frac{\pi}{4}\\
		&\leq
		\sqrt{1-\frac{1}{c'} }\cdot \frac{\pi}{4} =\sqrt{\frac{c'-1}{c'} }\cdot \frac{\pi}{4},
		\end{align*}
		where the second inequality follows since $ p,x\in B_{d} $, the third inequality follows since $f(z)=\frac{4z^2}{1-2z^2} $ is increasing for $ z\in(-\infty,\frac{1}{2}) $ and $ \eps<\frac{1}{4}$, the fifth inequality follows since $\eps\leq \sqrt{1-\frac{1}{\sqrt[4]{c}} }\cdot \frac{\pi}{8\sqrt{2}}$, and the last inequality follows from Inequality~(\ref{ineq:ctaglowerbound}).
		
		\item Assume that $\Edst{q_0}{x}\leq r$.
		We prove that $ \EAngdst{q}{\up{x}}\leq r' $. Indeed, by Lemma~\ref{lma:ellipsoidToAngularBound},
		\begin{align*}
		\EAngdst{q}{\up{x}}&= \sum_{i=1}^d w_i \cdot \angle_i(q,\up{x})^2\leq \sum_{i=1}^d w_i \cdot \eps^2(1+\eps)^2\left(|x_i-p_i|+\beta(\eps)\right)^2\\
		&=\eps^2(1+\eps)^2\cdot \sum_{i=1}^d w_i \cdot\left(|x_i-p_i|+\beta(\eps)\right)^2\\
		&\leq \eps^2(1+\eps)^2\cdot \sum_{i=1}^d w_i \cdot\left(|x_i-p_i|^2+\beta^2(\eps)+4\beta(\eps)\right)\\
		&= \eps^2(1+\eps)^2\cdot\sum_{i=1}^d w_i |x_i-p_i|^2\\
		&+\eps^2(1+\eps)^2\left(\sum_{i=1}^d w_i\right) \cdot\left(\beta^2(\eps)+4\beta(\eps)\right) \\
		&\leq \eps^2(1+\eps)^2r+\eps^2(1+\eps)^2\left(\sum_{i=1}^d w_i\right)\cdot\left(\beta^2(\eps)+4\beta(\eps)\right)\\
		&\leq \eps^2(1+\eps)^2r+5\eps^2(1+\eps)^2\left(\sum_{i=1}^d w_i\right)\beta(\eps)=r'.
		\end{align*}
		The second inequality follows since $ |x_i-p_i|\leq 2 $, the third inequality follows since $ \left(e_i^T(x-p)\right)^2= |x_i-p_i|^2$ and since we assumed that $\Edst{q_0}{x}\leq r$, and the last inequality follows since $ \beta(\eps)<1 $ for $ \eps\in [0,\frac{1}{8}] $.
		
		\item Assume that $\Edst{q_0}{x} > cr$. We prove that $\EAngdst{q}{x}> c'r'$. Indeed, denote $ I= \{i\mid |x_i-p_i|\geq \beta(\eps)\}\subseteq [n]$, and note that
		\begin{align*}
		cr&< \Edst{q_0}{x}=\sum_{i\in I} w_i \cdot|x_i-p_i|^2+\sum_{i\in [d]\setminus I} w_i \cdot|x_i-p_i|^2\\
		&\leq \sum_{i\in I} w_i \cdot|x_i-p_i|^2+\beta(\eps)^2 \cdot \sum_{i=1}^d w_i,
		\end{align*}
		
		and therefore
		\begin{equation}\label{eq:elpsdatleastcrminussomething}
		\sum_{i\in I} w_i \cdot|x_i-p_i|^2> cr- \beta(\eps)^2 \cdot \sum_{i=1}^d w_i.
		\end{equation}
		
		We define $ \phi =  \sum_{i\in I} w_i \cdot\left(|x_i-p_i|-\beta(\eps)\right)^2$, and observe that
		\begin{align}
		\phi&\geq \sum_{i\in I} w_i \cdot\left(|x_i-p_i|^2+\underbrace{\beta^2(\eps)}_{\geq 0}-4\beta(\eps)\right)
		\nonumber\\
		&
		\geq \sum_{i\in I} w_i \cdot|x_i-p_i|^2- 4\left(\sum_{i=1}^d w_i\right)\cdot\beta(\eps)\nonumber\\
		&\stackrel{(\ref{eq:elpsdatleastcrminussomething})}{>} cr-\beta(\eps)^2\sum_{i=1}^d w_i-4\beta(\eps)\sum_{i=1}^d w_i
		=cr-(4\beta(\eps)+\beta(\eps)^2) \sum_{i=1}^d w_i\nonumber\\
		&\geq cr-5\beta(\eps)\cdot \sum_{i=1}^d w_i.\label{eq:philowerbound}
		\end{align}
		The first inequality follows since $ |x_i-p_i|\leq 2 $, and the last inequality follows since $ \beta(\eps)<1 $ for $ \eps\in [0,\frac{1}{8}] $.
		
		Thus, by Lemma~\ref{lma:ellipsoidToAngularBound} we conclude that
		\begin{align*}
		\EAngdst{q}{\up{x}}&= \sum_{i=1}^d w_i \cdot \angle_i(q,\up{x})^2
		\geq \sum_{i\in I} w_i \cdot \angle_i(q,\up{x})^2\\
		&\geq \sum_{i\in I} w_i \cdot \left(\max \left(0,\eps(1-\eps)\cdot (|x_i-p_i|-\beta(\eps))\right)\right)^2\\
		& = \sum_{i\in I} w_i \cdot \eps^2(1-\eps)^2\cdot\left(|x_i-p_i|-\beta(\eps)\right)^2\\
		&=\eps^2(1-\eps)^2 \cdot \phi
		\stackrel{(\ref{eq:philowerbound})}{>}\eps^2(1-\eps)^2\cdot\left(cr-5\beta(\eps)\cdot \sum_{i=1}^d w_i\right)\geq c'r',
		\end{align*}
		where the second equality follows by the definition of $ I $, and the third equality follows by the definition of $ \phi $.
	\end{enumerate}
\end{proof}

Our reduction guarantees that any query $ q=(\up{p},\{\bar{e_i}\}_{i=1}^d) $ for the
angular ellipsoid structure and any data point $ \up{x}$ stored in it,
satisfy $ \angle(p,x)\leq \sqrt{\frac{c'-1}{c'} }\cdot \frac{\pi}{4}$ as required.

\subsection{From angular ellipsoid ALSH to weighted geometric angular similarity SLSH }\label{subsec:angellto1geomangsim}

In this section, we reduce the angular ellipsoid $(r,cr)$-\ALSH problem that we have studied in Section~\ref{subsec:eucelltoangell}, to a weighted geometric angular similarity $ (r',c'r') $-\SLSH problem.

\subsubsection{H-hash - the LSH scheme of Jain et al.}
Our data structure is based on the H-hash of Jain et al.~\cite{jainhashing}.
The H-hash stores points which reside on $ S_{d+1} $ such that
for a query hyperplane $h$ through the origin, we can efficiently
retrieve the data points that have a small angular distance with their projection on $h$.

H-hash in fact uses an SLSH family for the s2p geometric angular similarity for sets of size $2$. That is,
a hash function is defined by two random directions $u$ and $v$. We hash a
point $x$  to the concatenation of $sign(x^T u)$ and $sign(x^T v) $ and we represent a query
hyperplane $h$, perpendicular to  $ e $,
by the set  $ \{e,-e\} $, which is hashed to the concatenation of
$sign(e^T u)$ and $sign((-e)^T v)$.

The probability that a data point $x$ collides with the hyperplane  $h$ perpendicular to $e$ is
equal to $\angle sim (x,e)\cdot  \angle sim (x,-e) = (1-\angle(x,e)/\pi)(1-\angle(x,-e)/\pi)$.
This collision probability increases with the angle between $x$ and its projection on $h$, and attains its maximum when $x$ is on $h$.

Recall that the
angular ellipsoid distance between a query $ q=(p,\{e_i\}_{i=1}^d) $
and a point $x$ is a weighted sum of the terms $ (\angle_i(q,x))^2 $.
Therefore,  if we hash the hyperplane orthogonal to $ e_i $ with H-hash, it will collide with higher probability with data points $ x $ with a smaller $ \angle_i(q,x) $. This suggests that we can answer an angular ellipsoid query $q= (p,\{e_i\}_{i=1}^d) $
by  a weighted geometric angular similarity  SLSH  set-query where the set is the union of the sets $ \{e_i,-e_i\} $ for all $i\in [d] $, using an appropriate weight $ w_i $ for each $ i\in [d] $.
Specifically, given the parameters $ r>0$ and $c>1 $, we store the data points in an $ (S',c'S') $-\SLSH structure for the weighted geometric angular s2p similarity for queries of size $ k'=2d $ and with the weights $\{w_1,w_1,w_2,w_2,\ldots,w_d,w_d\} $.\footnote{Such a structure is given in Appendix~\ref{subsec:weightedgeom}.}
We define $ c' $ and $ S' $ as follows
{\[S' =e^{\sum_{i=1}^{d}w_i\cdot \ln\left(\frac{1}{4}\right) -\frac{ 4r}{\pi^2-4\psi_c^2}},\text{ and }c' = \frac{e^{\sum_{i=1}^{d}w_i\ln\left(\frac{1}{4}\right)-\frac{4cr}{\pi^2} }}{S'}= e^{-4r\left(\frac{c}{\pi^2}-\frac{1}{\pi^2-4\psi_c^2}\right)},\]}
where we define $ \psi_c =\sqrt{\frac{c-1}{c} }\cdot \frac{\pi}{4}$. \footnote{We will prove that $ c'<1 $.}
To answer an angular ellipsoid query $q= (p,\{e_i\}_{i=1}^d) $, we query our structure with the set-query $Q=\{e_1,-e_1,e_2,-e_2,\ldots,\\e_d,-e_d\} $. For the reduction to succeed, we require that any query $q= (p,\{e_i\}_{i=1}^d) $ and data point $ x $ satisfy $ \angle(p,x) \leq\sqrt{\frac{c-1}{c} }\cdot \frac{\pi}{4} $.

Correctness of our structure follows from the following two theorems.
\begin{theorem}\label{thm:probAntipodalPair}
	Let $x\in S_{d+1} $ and $q=(p,\{e_i\}_{i=1}^d)$ be an angular ellipsoid query. Then, $\angle sim_{geo}(\left\{e_i,-e_{i}\right\},x)= \frac{1}{4}-\frac{\angle_i(q,x)^2}{\pi^2}$ for all $i\in [d] $.
\end{theorem}
\begin{proof}
	Let $x\in S_{d+1},~ i\in [d] $ and an angular ellipsoid query $q=(p,\{e_i\}_{i=1}^d)$. Recall that $ \angle_i(q,x) $ is defined as the angle between $ x $ and its projection on the hyperplane perpendicular to $ e_i $ passing through the origin, i.e., $\angle_i(q,x)= \frac{\pi}{2}-\min\left(\angle(x,e_i),\angle(x,-e_i)\right) $.
	Rearranging, we get that $ \min\left(\angle(x,e_i),\angle(x,-e_i)\right) =  \frac{\pi}{2}-\angle_i(q,x)$ , and since $ \angle(x,e_i)+\angle(x,-e_i)=\pi $, we get that $\max\left(\angle(x,e_i),\angle(x,-e_i)\right) =  \frac{\pi}{2}+\angle_i(q,x) $. So
	we get that the set of angles $ \{\angle(x,e_i),\angle(x,-e_i)\} $ is equal to $\{ \frac{\pi}{2}-\angle_i(q,x),\frac{\pi}{2}+\angle_i(q,x)\}  $ (see Figure~\ref{fig:thetai}). Hence, from the definition of the geometric angular similarity we get 
	\begin{align*} 
	\angle sim_{geo}(\left\{e_i,-e_{i}\right\},x)
	&=\left(1-\frac{\angle(x,e_i)}{\pi}\right)\cdot \left(1-\frac{\angle(x,-e_i)}{\pi}\right)\\
	&= \left(1-\frac{\frac{\pi}{2}-\angle_i(q,x)}{\pi}\right)\cdot \left(1-\frac{\frac{\pi}{2}+\angle_i(q,x)}{\pi}\right)\\
	&=\frac{1}{4}-\frac{\angle_i(q,x)^2}{\pi^2}.
	\end{align*}
\end{proof}

\begin{theorem}\label{thm:angellto1geomangsim}
	The structure above is an $ (r,cr) $-\ALSH structure for the angular ellipsoid distance $ \EAngd $.
\end{theorem}
\begin{proof}
	Consider an angular ellipsoid query $ q= (p, \{e_i\}_{i=1}^d)$ for the angular ellipsoid structure, and let $ Q =\{e_1,-e_1,e_2,-e_2,\ldots,e_d,-e_d\}$ be the corresponding weighted set-query for the weighted geometric angular similarity $(S',c'S')$-\SLSH structure with the weights $\{w_1,w_1,w_2,w_2,\ldots,w_d,w_d\} $, as defined in the query phase. It suffices to prove that
	\begin{enumerate}
		\item $ c'<1$,
		\item $\forall x~s.t.~\EAngdst{q}{x}\leq r,~ \angle sim_{wgeo}(Q,x)\geq S'$, and
		\item $\forall x~s.t.~\EAngdst{q}{x}> cr,~\angle sim_{wgeo}(Q,x)< c'S'$.
	\end{enumerate}
	The proofs of these claims are as follows.
	\begin{enumerate}
		\item By the definition of $ c' $ we get that $\frac{\ln(c')}{-4r}=\frac{c}{\pi^2}  - \frac{1 }{\pi^2-4\psi_c^2}>\frac{c}{\pi^2}  - \frac{1 }{\pi^2-4\cdot \frac{c-1}{c}\cdot \frac{\pi^2}{4}}=0$, where the second equality follows since $\psi_c= \sqrt{\frac{c-1}{c} }\cdot \frac{\pi}{2} $. Thus, $ \ln(c')<0 $ so $ c'<1 $.
		\item Assume that $\EAngdst{q}{x}\leq r$. We prove that $ \angle sim_{wgeo}(Q,x)\geq S' $. Indeed, by Theorem~\ref{thm:probAntipodalPair} and since $ \forall i,~ \angle_i(q,x)\leq \angle(p,x)\leq\psi_c<\pi/2 $,\footnote{Here we use our assumption that any query $q= (p,\{e_i\}_{i=1}^d) $ and data point $ x $ in the angular ellipsoid structure satisfy $ \angle(p,x) \leq\sqrt{\frac{c-1}{c} }\cdot \frac{\pi}{4}=\psi_c $.} we have that
		\begin{align*}
		\angle sim_{wgeo}(Q,x)&=
		\prod_{i=1}^{d}\left(\angle sim(e_i,x)\right)^{w_i}\cdot \prod_{i=1}^{d}\left(\angle sim(-e_i,x)\right)^{w_i}\\
		&=\prod_{i=1}^{d}\left(\angle sim_{geo}(\{e_i,-e_i\},x)\right)^{w_i}\\
		&=\prod_{i=1}^{d}\left( \frac{1}{4}-\frac{\angle_i(q,x)^2}{\pi^2}\right)^{w_i}
		=e^{\sum_{i=1}^{d}w_i\cdot \ln\left(\frac{1}{4}-\frac{1}{\pi^2}\cdot \angle_i(q,x)^2\right)}\\
		&\geq e^{\sum_{i=1}^{d}w_i\cdot \left(\ln\left(\frac{1}{4}\right)-\frac{1}{\frac{1}{4}-\frac{1}{\pi^2}\cdot \angle_i(q,x)^2}\cdot \frac{\angle_i(q,x)^2}{\pi^2} \right)}\\
		&=e^{\sum_{i=1}^{d}w_i \left(\ln\left(\frac{1}{4}\right)-\frac{ 4\angle_i(q,x)^2}{\pi^2-4\angle_i(q,x)^2}\right)}
		\geq e^{\sum_{i=1}^{d}w_i \left(\ln\left(\frac{1}{4}\right)-\frac{ 4\angle_i(q,x)^2}{\pi^2-4\psi_c^2}\right)}\\
		&=e^{\sum_{i=1}^{d}w_i\cdot \ln\left(\frac{1}{4}\right) -\frac{ 4\EAngdst{q}{x}}{\pi^2-4\psi_c^2}}
		\geq e^{\sum_{i=1}^{d}w_i\cdot \ln\left(\frac{1}{4}\right) -\frac{ 4r}{\pi^2-4\psi_c^2}}=S',
		\end{align*}
		where the first and second equalities follow by the definition of the geometric and weighted geometric similarities, the first inequality follows since
		$\ln\left(\frac{1}{4}\right)-\ln\left(\frac{1}{4}-\frac{\angle_i(q,x)^2}{\pi^2} \right)=\frac{1}{\mu}\cdot \frac{\angle_i(q,x)^2}{\pi^2}$
		for some $\mu\in [\frac{1}{4}-\frac{\angle_i(q,x)^2}{\pi^2} ,\frac{1}{4}] $ by Lagrange's mean value theorem, and the fourth equality follows since $\pi^2-4\angle_i(q,x)^2\geq  \pi^2-4\psi_c^2\geq 0$.
		
		\item Assume that $\EAngdst{q}{x} > cr$. We prove that $\angle sim_{wgeo}(Q,x) < c'S'$. Indeed, as in the proof of claim 2,
		\begin{align*}
		\angle sim_{wgeo}(Q,x)&
		=e^{\sum_{i=1}^{d}w_i\cdot \ln\left(\frac{1}{4}-\frac{1}{\pi^2}\cdot \angle_i(q,x)^2\right)}
		\\&\leq e^{\sum_{i=1}^{d}w_i\cdot \left(\ln\left(\frac{1}{4}\right)-\frac{1}{1/4}\cdot \frac{\angle_i(q,x)^2}{\pi^2} \right)}\\
		&=e^{\sum_{i=1}^{d}w_i\ln\left(\frac{1}{4}\right)-\frac{4\EAngdst{q}{x}}{\pi^2} }
		< e^{\sum_{i=1}^{d}w_i\ln\left(\frac{1}{4}\right)-\frac{4cr}{\pi^2} }=c'S',
		\end{align*}
		where the first inequality follows since
		$\ln\left(\frac{1}{4}\right)-\ln\left(\frac{1}{4}-\frac{\angle_i(q,x)^2}{\pi^2} \right)=\frac{1}{\mu}\cdot \frac{\angle_i(q,x)^2}{\pi^2}$
		for some $\mu\in [\frac{1}{4}-\frac{\angle_i(q,x)^2}{\pi^2} ,\frac{1}{4}] $ by Lagrange's mean value theorem.
	\end{enumerate}
\end{proof}

\section{Missing proofs from Section~\ref{sec:cedforquerysize2}}
\begin{proof}[Proof of Lemma~\ref{lma:ellipsoidsBetweenConstantHeights}]
	Note that all the sets $ L^s,~S,~B$ and $L^b$ are spherically symmetric around the axis $ x_1 $ and symmetric with respect to the hyperplane $ x_1=0 $.
	We denote the distances of the boundaries of $ L^s,~S,~B$ and $L^b$ from the axis $ x_1 $ as functions of $ x_1 $ by $y_{L^s}(x_1),~y_{S}(x_1),~y_{B}(x_1)$ and $ y_{L^b}(x_1)$ respectively, which are plotted in Figure~\ref{fig:ellipsoidHierarchy}.
	
	Note that $ y_{L^s}$ and $y_S $ are defined only over the domain $|x_1|\leq r-a $, $ y_B $ is defined only over the domain $ |x_1|\leq  d_B$ for some $ d_B\in[r-a,cr-a) $,\footnote{$ d_B\geq r-a $, since if we substitute $ x_1 $ with $ r-a $ in the equation of $ B $, we get $ r^2-a^2+\sum_{i=2}^{d}x_i^2 \leq  \left(\frac{cr}{c_{\min}}\right)^2-a^2$, which holds for example for $ x_i=0 $, $ i=2,\ldots,d $, since $ c> c_{\min} $. $ d_B\leq cr-a $, since if we substitute $ x_1 $ with $ cr-a $ in the equation of $ B $, we get $ \frac{r+a}{r-a}\cdot \left(cr-a\right)^2+\sum_{i=2}^{d}x_i^2 \leq  \left(\frac{cr}{c_{\min}}\right)^2-a^2$. There are no $ x_i $'s $ i=2,\ldots,d $ that satisfy this inequality, since $ \frac{r+a}{r-a}\cdot \left(cr-a\right)^2 > \frac{cr+a}{cr-a} \cdot \left(cr-a\right)^2=\left(cr\right)^2-a^2>\left(\frac{cr}{c_{\min}}\right)^2-a^2$, where the first inequality follows since
		$ \frac{z+a}{z-a} $ is a strictly decreasing function for $ z>a $, and $ c>1$, and the last inequality follows since $ c_{\min}>1 $.} 
	and $ y_{L^b} $ is defined only over the domain $ |x_1|\leq cr-a $. Therefore, using the fact that $y_{L^s}(x_1),~y_{S}(x_1),~y_{B}(x_1)$ and $ y_{L^b}(x_1)$ are non-negative, it suffices to prove the following two claims.
	\begin{enumerate}
		\item $\forall x\in[0, r-a],~y_{L^s}^2(x)\leq y_{S}^2(x)\leq y_{B}^2(x)$,
		\item $\forall x\in[0, d_B],~ y_{B}^2(x) \leq y_{L^b}^2(x)$
	\end{enumerate}
	Since all points $ v=(x_1,\ldots,x_d) $ with $ x_1>0 $ are farther away from $ q_{-a} $ than $ q_a $, the intersections of the boundaries of $ L^s,~S,~B$ and $L^b$ with the hyperplane $ x_1=x $ for $ x>0 $ are defined respectively by the equations
	\begin{align*}
	&L^s[x]:~ \sqrt{(x+a)^2+y_{L^s}^2(x)} = r,\\
	&S[x]:~ \frac{r+a}{r-a}x^2+y_{S}^2(x) = r^2-a^2,\\
	&B[x]:~ \frac{r+a}{r-a}x^2+y_{B}^2(x) =  \left(\frac{cr}{c_{\min}}\right)^2-a^2\text{, and}\\
	&L^b[x]:~ \sqrt{(x+a)^2+y_{L^b}^2(x)} = c r.
	\end{align*}
	From these equations we get that
	
	\begin{enumerate}[i.]
		\item $ y_{L^s}^2(x)=r^2-(x+a)^2$\label{eq:yLs},
		\item $y_{S}^2(x) = r^2-a^2- \frac{r+a}{r-a}x^2 $\label{eq:yS},
		\item $y_{B}^2(x) =  \left(\frac{cr}{c_{\min}}\right)^2-a^2-\frac{r+a}{r-a}x^2$\label{eq:yB}, and
		\item $y_{L^b}^2(x) = (c r)^2- (x+a)^2$\label{eq:yLb}.
	\end{enumerate}
	
	We now establish claims 1 and 2 in order.
	\begin{enumerate}
		\item We show that $\forall x\in[0, r-a],~ y_{L^s}^2(x)\leq y_{S}^2(x)\leq y_{B}^2(x) $. Let $ x\in[0, r-a] $, and observe that
		\begin{align*}
		y_{S}^2(x)&\stackrel{\ref{eq:yS}}{=}r^2-a^2-\frac{r+a}{r-a}x^2
		=r^2-(x+a)^2-\frac{2a}{r-a}x^2+2ax\\
		&=r^2-(x+a)^2+\frac{2ax}{r-a}\underbrace{(r-a-x)}_{\geq 0}\geq r^2-(x+a)^2\stackrel{\ref{eq:yLs}}{=}y_{L^s}^2(x).
		\end{align*}
		Moreover,
		\begin{align*}
		y_{S}^2(x)&\stackrel{\ref{eq:yS}}{=}r^2-a^2-\frac{r+a}{r-a}x^2\leq \left(\frac{cr}{c_{\min}}\right)^2-a^2-\frac{r+a}{r-a}x^2\stackrel{\ref{eq:yB}}{=}y_{B}^2(x),
		\end{align*}
		where the first inequality follows since $  c>c_{\min} $.
		
		\item We prove the stronger claim that $\forall x,~ y_{B}^2(x) \leq y_{L^b}^2(x)$. Indeed, for any $ x $ observe that
		\begin{align*}
		y_{L^b}^2(x)-y_{B}^2(x)&\stackrel{\ref{eq:yLb},\ref{eq:yB}}{=} (c r)^2- (x+a)^2-\left(\left(\frac{cr}{c_{\min}}\right)^2-a^2-\frac{r+a}{r-a}x^2\right)\\
		&=(cr)^2- x^2-2ax -\left(\frac{cr}{c_{\min}}\right)^2+\frac{r+a}{r-a}x^2\\
		&=\left(\frac{r+a}{r-a}-1\right)x^2-2ax+c^2\left(1-1/c^2_{\min}\right) r^2\\
		&>\left(\frac{r+a}{r-a}-1\right)x^2-2ax+c^2_{\min}\left(1-1/c^2_{\min}\right) r^2\\
		&=\frac{2a}{r-a}x^2-2ax+\left(c^2_{\min}-1\right) r^2\\
		&=\frac{2a}{r-a}x^2-2ax+\frac{1}{8} r^2\\
		&=\frac{2}{r-a}\cdot \left(ax^2-a(r-a)x+\frac{1}{16}\left(r^3-ar^2\right)\right),
		\end{align*}
		where the inequality follows since $c> c_{\min} $ and $ c_{\min}>1 $, and the fifth inequality follows since $ c_{\min}=\frac{3}{2\sqrt{2}} $.
		It remains to show that the second term in the previous equation is non-positive for all $ x$, i.e., to show that $ \forall x,~ t(x):=ax^2-a(r-a)x+\frac{1}{16}(r^3-ar^2)\geq 0 $. We divide $ t(x) $ by $ a> 0 $ and show that $  \forall x\in ,~\gamma(x):=x^2-(r-a)x+\frac{1}{16a}(r^3-ar^2)\geq 0 $.
		To do so, we show that the discriminant $ \Delta(\gamma)=(r-a)^2-\frac{1}{4a}(r^3-ar^2) = (r-a)\left(r-a-\frac{1}{4a}r^2\right)=-\frac{1}{4a}\cdot (r-a) (r^2 - 4a r+4a^2)= -\frac{1}{4a}\cdot (r-a) (r-2a)^2$ is non-positive. Indeed, we have that $ r\geq a $ and $ a>0 $, and we finish.
	\end{enumerate}
\end{proof}

\begin{proof}[Proof of Theorem~\ref{thm:1centerToEuclideanEllipsoidReductionSLSH}]
	Let $\tilde{Q}= \{\tilde{q_1},\tilde{q_2}\} $ be a set-query, and let $ a=\frac{1}{2} \norm{\tilde{q_1}-\tilde{q_2}} \in [0,(1-\phi)r)$.
	By the definition of the structure, we make an ellipsoid query $ (\tilde{p},\{\tilde{e}_i\}_{i=1}^d) $ with $ \tilde{p}=\frac{1}{2}(\tilde{q_1}+\tilde{q_2}) $, and the first axis $ \tilde{e}_1 $ is the unit vector in the direction of $ \tilde{q_1}-\tilde{q_2} $, and all other axes complete $ \tilde{e}_1 $ to an orthonormal basis.
	For the rest of the proof, we change the coordinate system such that it is centered in $ \tilde{p} $ and the directions of the axes are $\{\tilde{e}_i\}_{i=1}^d$. In this system, the set-query is $ \{q_a,q_{-a}\} $ where $ q_a=(a,0\ldots,0) $, $ q_{-a}=(-a,0\ldots,0)$, and we query the euclidean ellipsoid $(r',c'r')$-\ALSH structure with $ ((0,\ldots,0),\{e_i\}_{i=1}^d) $, where $ \{e_i\}_{i=1}^d $ is the standard basis. Moreover, a point $ \tilde{v} $ in the original system, is represented by $v= (x_1,\ldots,x_d) $ in the new system, where $ x_i=e_i^T(\tilde{v}-\tilde{p})  $ for all $ i\in [d] $.
	It therefore suffices to show that all data points $ v $ for which
	$ \max \left(\norm{v-q_a},\norm{v-q_{-a}}\right)\leq r $ are within $ r' $ euclidean ellipsoid distance from the query $ q $, and all data points $ v $ for which $ \max \left(\norm{v-q_a},\norm{v-q_{-a}}\right)> cr $ are at euclidean ellipsoid distance strictly larger than $ c'r' $  from the query $ q $.
	
	Specifically, let $ a'=\left\lceil\frac{a}{\delta}\right\rceil \cdot \delta $, let
	\begin{align*}
	&S^+=\left\{(x_1,\ldots,x_d) \mid \frac{r+a'}{r-a'}x_1^2+\sum_{i=2}^{d}x_i^2 \leq \frac{c}{c_{\min}}\cdot \left(r^2-(a')^2\right)\right\}, \text{ and let }\\
	&B^-=\left\{(x_1,\ldots,x_d) \mid \frac{r+a'}{r-a'}x_1^2+\sum_{i=2}^{d}x_i^2 \leq   \left(\frac{c}{c_{\min}}\right)^2\cdot \left(r^2-(a')^2\right)\right\}
	\end{align*}
	as defined in Section~\ref{sec:cedforquerysize2}. We prove that $L^s\subseteq S^+$ and $B^-\subseteq L^b$.
	By Lemma~\ref{lma:ellipsoidsBetweenConstantHeights}, $ L^s\subseteq S $ and $ B\subseteq L^b $, so it suffices to show that
	\begin{enumerate}
		\item $S\subseteq S^+$, and
		\item $B^-\subseteq B$.
	\end{enumerate}
	We now prove these two claims one after another.
	\begin{enumerate}
		\item Fix $ v=(x_1,\ldots,x_d) $ such that $ v \in S$, i.e., $\frac{r+a}{r-a}x_1^2+\sum_{i=2}^d x_i^2\leq r^2-a^2$. We show that $ v\in S^+ $. Indeed,
		\begin{align*}
		\frac{r+a'}{r-a'} x_1^2+\sum_{i=2}^d x_i^2
		&= \left(\frac{r+a'}{r-a'}\cdot \frac{r-a}{r+a}\right)\cdot \frac{r+a}{r-a}x_1^2+\sum_{i=2}^d x_i^2\\
		&\leq  \left(\frac{r+a'}{r-a'}\cdot \frac{r-a}{r+a}\right)\cdot \frac{r+a}{r-a}x_1^2+\left(\frac{r+a'}{r-a'}\cdot \frac{r-a}{r+a}\right)\cdot \sum_{i=2}^d x_i^2\\
		&=\left(\frac{r+a'}{r-a'}\cdot \frac{r-a}{r+a}\right)\cdot \left(\frac{r+a}{r-a}x_1^2+\sum_{i=2}^d x_i^2\right)\\
		&\leq \left(\frac{r+a'}{r-a'}\cdot \frac{r-a}{r+a}\right)\cdot \left(r^2-a^2\right)\\
		&=\frac{(r+a')\cdot (r-a)\cdot (r+a)\cdot (r-a)}{(r-a')\cdot (r+a)}\\
		&=\frac{(r+a')\cdot (r-a)^2}{(r-a')}=\frac{(r-a)^2}{(r-a')^2}\cdot (r^2-(a')^2)\\
		&\leq \left(\frac{r-a}{r-a-\delta}\right)^2\cdot (r^2-(a')^2)
		\leq \left(\frac{\phi r}{\phi r-\delta}\right)^2\cdot (r^2-(a')^2)\\
		&\leq \frac{c}{c_{\min}}\cdot \left(r^2-(a')^2\right).
		\end{align*}
		The first inequality follows since $ \frac{r+x}{r-x} $ is an increasing function for $x\in [0,r]  $ and $ a'\geq a\in [0,r] $, and the second inequality follows since $ v\in S $.
		The third inequality follows since $ r-a'\geq r-a-\delta\geq 0 $ as we argued in Section~\ref{sec:cedforquerysize2}, and the fourth inequality follows since $ \frac{x}{x-\delta} $ is decreasing for $ x>\delta $, since $ r-a\geq \phi r $, and since $ \phi r\geq \delta $.
		The last inequality follows since $ \frac{\phi r}{\phi r-\delta}\leq \sqrt{\frac{c}{c_{\min}}} $ because $ \delta \leq \left(1-\sqrt{\frac{c_{\min}}{c}}\right)\cdot \phi r $.
		
		\item Fix $ v=(x_1,\ldots,x_d) $ such that $ v\notin B $, i.e., $\frac{r+a}{r-a}x_1^2+\sum_{i=2}^d x_i^2>\left( \frac{cr}{c_{\min}}\right)^2-a^2$.We show that $ v\notin B^- $. Indeed,
		\begin{align*}
		\frac{r+a'}{r-a'} x_1^2+\sum_{i=2}^d x_i^2
		&\geq \frac{r+a}{r-a} x_1^2+\sum_{i=2}^d x_i^2
		> \left( \frac{cr}{c_{\min}}\right)^2-a^2
		\geq \left( \frac{cr}{c_{\min}}\right)^2-(a')^2\\
		&\geq \left( \frac{c}{c_{\min}}\right)^2\left(r^2-(a')^2\right).
		\end{align*}
		The first inequality follows since $ \frac{r+x}{r-x} $ is an increasing function for $x\in [0,r]  $ and $ a'\geq a $, the second inequality follows since $  v\notin B $, the third inequality follows since $ a'\geq a $, and the last inequality follows since $ c> c_{\min} $.
	\end{enumerate}
\end{proof}

\end{document}